\documentclass[final,nocopyrightspace]{sigplanconf}

\usepackage[english]{babel}

\usepackage{amsmath}
\usepackage{amsfonts,mathrsfs,amssymb}
\usepackage{amsthm}
\usepackage{stmaryrd}

\usepackage{multirow}

\usepackage[inline]{enumitem}
\setlist[itemize]{label=-}

\setlist[enumerate]{label=\small(\roman*), labelwidth=3pt}

\newlist{asparaenum}{enumerate}{3}
\setlist[asparaenum]{label=(\roman*), align=left, leftmargin=0pt, labelindent=!, listparindent=\parindent, labelwidth=0pt, itemindent=!}

\newlist{inparaenum}{enumerate*}{3}
\setlist[inparaenum]{label=(\roman*)}

\newlist{asparaitem}{enumerate}{3}
\setlist[asparaitem]{label=-, align=left, leftmargin=0pt, labelindent=!, listparindent=\parindent, labelwidth=0pt, itemindent=!}

\usepackage{booktabs}

\theoremstyle{plain}
\newtheorem{lemma}{Lemma}
\newtheorem{theorem}{Theorem}

\newtheorem{corollary}{Corollary}
\theoremstyle{definition}
\newtheorem{definition}{Definition}
\newtheorem{example}{Example}

\newtheorem*{correctness-assumption}{Correctness Assumption}
\theoremstyle{remark}
\newtheorem{remark}{Remark}

\newenvironment{proof*}
    {\begin{proof}[Proof \textup(Sketch\textup)]}
    {\end{proof}}

\usepackage[ final,
             bookmarks,
             bookmarksopen,
             colorlinks,
             breaklinks,
             linkcolor=blue,
             citecolor=black!50!brown,
             pdfstartview=FitH ]%
{hyperref}
\usepackage[hyphenbreaks]{breakurl}

\input{macros}
\usepackage[final]{listings}

\usepackage{xcolor}
\colorlet{keyword}{blue!50!black}
\colorlet{atom}{red!50!black}
\colorlet{module}{green!30!black}
\colorlet{comment}{black!70}
\colorlet{coderules}{black!50}
\colorlet{lineno}{black!50}

\makeatletter
\newcommand{\srcsize}{\@setfontsize{\srcsize}{8pt}{8pt}}
\makeatother

\lstdefinelanguage{CoreErlang}
  {morekeywords={fun,and,case,letrec,let,in,catch,div,end,exit,export,halt,%
      if,import,link,make_ref,module,monitor,of,or,receive,self,send,spawn,throw,to,%
      unlink%
      },%
   morekeywords={[2]error,false,nil,ok,true,undefined,pid},%
   otherkeywords={!},%
   morecomment=[l]\%,%
   morestring=[b]"%
  }[keywords,comments,strings]%

\lstdefinestyle{sans}{
    xleftmargin=20pt,
    tabsize=4,
    showstringspaces=false,
    columns=[l]flexible,
    breaklines,
    fontadjust,
    numbers=left,
    numberstyle={\tiny\color{lineno}},
    literate={
        {->}{{$\hspace{-1pt}\rightarrow$}}2
        {...}{{$\ldots$}}3
    },
    basicstyle={\sffamily\srcsize}, % \ttfamily
    keywordstyle={\color{keyword}\bf},
    keywordstyle={[2]\color{atom}},
    commentstyle={\scriptsize\color{comment}},
    emphstyle={\color{atom}},
    emphstyle={[2]\color{module}},
    moredelim=[is][emphstyle]{\#}{\#},
    moredelim=[is][emphstyle2]{@}{@},
    escapechar=§,mathescape
}
\lstdefinestyle{boxed}{frame=single,frameround=tttt,backgroundcolor=\color{keyword!5}}
\lstdefinestyle{head}{style=boxed, frame=tlr, frameround=tfft}
\lstdefinestyle{middle}{style=boxed, frame=lr, firstnumber=last}
\lstdefinestyle{tail}{style=boxed, frame=lrb, frameround=fttf, firstnumber=last}

\lstdefinestyle{inl}{columns=[c]flexible, basicstyle={\sffamily\footnotesize}}

\newcommand{\erl}[2][]{\ifmmode\expandafter\text\fi{\lstinline[style=inl,#1]{#2}}}

\lstnewenvironment{erlang}[1][]{
 \lstset{language=CoreErlang, style=sans, #1}
}{}
\lstnewenvironment{erlang*}[1][]{
	\lstset{language=CoreErlang, style=sans, numbers=none, #1}
}{}
\lstnewenvironment{erlango}[1][]{
  \lstset{language=CoreErlang, style=sans, style=boxed, frame=tlr,%
  frameround=ttff, #1} }{}

\allowdisplaybreaks

\usepackage{tikz}
\usetikzlibrary{shapes,arrows,positioning,fit,matrix,automata}

\begin{document}

\title{Automatic Verification of Erlang-Style Concurrency}

\authorinfo{Emanuele D'Osualdo}{University of Oxford}{emanuele.dosualdo@cs.ox.ac.uk}
\authorinfo{Jonathan Kochems}{University of Oxford}{jonathan.kochems@cs.ox.ac.uk}
\authorinfo{C.-H.~Luke~Ong}{University of Oxford}{luke.ong@cs.ox.ac.uk}

\maketitle

\begin{abstract}
This paper presents an approach to verify safety properties of Erlang-style,
higher-order concurrent programs \emph{automatically}. Inspired by {Core
Erlang}, we introduce \Lang, a prototypical functional language with
pattern-matching algebraic data types, augmented with process creation and
asynchronous message-passing primitives. We formalise an abstract model of
\Lang\ programs called \emph{Actor Communicating System} (ACS) which has a
natural interpretation as a vector addition system, for which some verification
problems are decidable. We give a parametric abstract interpretation framework
for \Lang{} and use it to build a polytime computable, flow-based, abstract
semantics of \Lang\ programs, which we then use to bootstrap the ACS
construction, thus deriving a more accurate abstract model of the input program.
We have constructed \Soter, a tool implementation of the verification method,
thereby obtaining the first fully-automatic, infinite-state model checker for a
core fragment of Erlang. We find that in practice our abstraction technique is
accurate enough to verify an interesting range of safety properties. Though the
ACS coverability problem is \expspace-complete, \Soter\ can analyse these
verification problems surprisingly efficiently.

\end{abstract}

%\category{D.2.4}{Software Engineering}{Software/Program Verification}

% !TEX root = main.tex
\section{Introduction}\label{sec:intro}

This paper concerns the verification of concurrent programs written in Erlang.
Originally designed to program fault-tolerant distributed systems at Ericsson in
the late 80s, Erlang is now a widely used, open-sourced language with support
for higher-order functions, concurrency, communication, distribution, on-the-fly
code reloading, and multiple platforms~\cite{Armstrong:93,Armstrong:10}.
Largely because of a runtime system that offers highly efficient process
creation and message-passing communication, Erlang is a natural fit for
programming multicore CPUs, networked servers, parallel databases, GUIs, and
monitoring, control and testing tools.

The sequential part of Erlang is a higher order, dynamically typed,
call-by-value functional language with pattern-matching algebraic data types.
Following the \emph{actor model}~\cite{Agha:86}, a concurrent Erlang
computation consists of a dynamic network of processes that communicate by
message passing. Every process has a unique process identifier (pid), and is
equipped with an unbounded mailbox. Messages are sent asynchronously in the
sense that send is non-blocking. Messages are retrieved from the mailbox, not
FIFO, but First-In-First-Firable-Out (FIFFO) via pattern-matching. A process may
block while waiting for a message that matches a certain pattern to arrive in
its mailbox. For a quick and highly readable introduction to Erlang, see
Armstrong's \emph{CACM} article~\cite{Armstrong:10}.

\paragraph{Challenges.}

Concurrent programs are hard to write. They are just as hard to verify. In the case of Erlang programs, the inherent complexity of the verification task can be seen from several diverse sources of infinity in the state space.
\begin{enumerate}[label={($\infty$\,\arabic*)},leftmargin=3.5em]
    \item \label{infty:nontail}
        General recursion requires a (process local) call-stack.

    \item \label{infty:order}
        Higher-order functions are first-class values;
        closures can be
        passed as parameters or returned.

	\item \label{infty:domain}
		Data domains, and hence the message space, are unbounded: functions may return, and variables may be bound to, terms of an arbitrary size.

	\item \label{infty:spawn}
		An unbounded number of processes can be spawned dynamically.

	\item \label{infty:mailbox}
		Mailboxes have unbounded capacity.

\end{enumerate}
The challenge of verifying Erlang programs is that one must reason about the asynchronous communication of an unbounded set of messages, across an unbounded set of Turing-powerful processes.

Our goal is to verify safety properties of Erlang-like programs \emph{automatically}, using a combination of static analysis and infinite-state model checking. To a large extent, the key decision of which causes of infinity to model as accurately as possible and which to abstract is forced upon us: the class consisting of a fixed set of context-free (equivalently, first-order) processes, each equipped with a mailbox of size one and communicating messages from a finite set, is already Turing powerful \cite{DOsualdo:11}. Our strategy is thus to abstract \ref{infty:nontail}, \ref{infty:order} and \ref{infty:domain}, while seeking to analyse message-passing concurrency, assuming \ref{infty:spawn} and \ref{infty:mailbox}.

We consider programs of \Lang, a prototypical functional language with actor-style concurrency. \Lang\ is essentially \emph{Core Erlang} \cite{Carlsson:01}---the official intermediate representation of Erlang code, which exhibits in full the higher-order features of Erlang, with asynchronous message-passing concurrency and dynamic process creation.

With decidable infinite-state model checking in mind, we introduce \emph{Actor Communicating System} (ACS), which models the interaction of an unbounded set of communicating processes. An ACS has a finite set of control states $Q$, a finite set of \emph{pid classes} $P$, a finite set of messages $M$, and a finite set of transition rules. An ACS transition rule has the shape $\ecsrule[\pid]{q}{\ell}{q'}$, which means that a process of pid class $\pid$ can transition from state $q$ to state $q'$ with (possible) \emph{communication side effect} $\ell$, of which there are four kinds, namely,
\begin{inparaenum}
\item the process makes an internal transition
\item it extracts and reads a message $m$ from its mailbox
\item it sends a message $m$ to a process of pid class $\pid'$
\item it spawns a process of pid class $\pid'$.
\end{inparaenum}
ACS models are infinite state: the mailbox of a process has unbounded capacity,
and the number of processes in an ACS may grow arbitrarily large. However the
set of pid classes is fixed, and processes of the same pid class are not
distinguishable.

An ACS can be interpreted naturally as a \emph{vector addition system} (VAS), or
equivalently Petri net, using counter abstraction. Recall that a VAS of
dimension $n$ is given by a set of $n$-long vectors of integers regarded as
transition rules. A VAS defines a state transition graph whose states are just
$n$-long vectors of non-negative integers. There is a transition from state
$\vec{v}$ to state $\vec{v}'$ just if $\vec{v}' = \vec{v} + \vec{r}$ for some
transition rule $\vec{r}$. It is well-known that the decision problems
\emph{Coverability} and \emph{LTL Model Checking} for VAS are EXPSPACE-complete;
\emph{Reachability} is decidable but its complexity is open.
We consider a particular counter abstraction of ACS, called \emph{VAS
semantics}, which models an ACS as a VAS of dimension $|P| \times (|Q| + |M|)$,
distinguishing two kinds of counters. A counter named by a pair $(\pid, q)$
counts the number of processes of pid class $\pid$ that are currently in state
$q$; a counter named by $(\pid, m)$ counts the sum total of occurrences of a
message $m$ currently in the mailbox of $p$, where $p$ ranges over processes of
pid class $\pid$.
Using this abstraction, we can conservatively decide properties of the ACS using
well-known decision procedures for VAS.

\paragraph{Parametric, Flow-based Abstract Interpretation.}
The starting point of our verification pathway is the abstraction of the sources
of infinity \ref{infty:nontail}, \ref{infty:order} and \ref{infty:domain}.
Methods such as \kCFA\ \cite{Shivers:91} can be used to abstract higher-order
recursive functions to a finite-state system. Rather than `baking in' each type
of abstraction separately, we develop a general abstract interpretation
framework which is \emph{parametric} on a number of basic domains. In the style
of Van Horn and Might \cite{VanHorn:10}, we devise a machine-based operational
semantics of \Lang\ that uses \emph{store-allocated continuations}. The
advantage of such an indirection is that it enables the construction of a
machine semantics which is `generated' from the basic domains of \Time,
\Mailbox\ and \Data. We show that there is a simple notion of \emph{sound
abstraction of the basic domains} whereby every such abstraction gives rise to a
sound abstract semantics of \Lang\ programs (Theorem~\ref{thm:CFAsound}).
Further if a given sound abstraction of the basic domains is finite and the
associated auxiliary operations are computable, then the derived abstract
semantics is finite and computable.

\paragraph{Generating an Actor Communicating System.}
We study the abstract semantics derived from a particular \zCFA-like abstraction of the basic domains. However we do not use it to verify properties of \Lang\ programs directly, as it is too coarse an abstraction to be useful. Rather, we show that a sound ACS (Theorem~\ref{thm:ACSsound}) can be constructed in polynomial time by \emph{bootstrapping} from the \zCFA-like abstract semantics. Further, the dimension of the resulting ACS is polynomial in the length of the input \Lang\ program. The idea is that the \zCFA-like abstract (transition) semantics constitutes a sound but rough analysis of the control-flow of the program, which takes higher-order computation into account but communicating behaviour only minimally. The bootstrap construction consists in constraining these rough transitions with guards of the form `receive a message of this type' or `send a message of this type' or `spawn a process', thus resulting in a more accurate abstract model of the input \Lang\ program in the form of an ACS.

\paragraph{Evaluation.}
To demonstrate the feasibility of our verification method, we have constructed a
prototype implementation called \Soter. Our empirical results show that the
abstraction framework is accurate enough to verify an interesting range of
safety properties of non-trivial Erlang programs.

\paragraph{Outline.}
In Section~\ref{sec:erlang} we define the syntax of \Lang\ and informally
explain its semantics with the help of an example program.
In Section~\ref{sec:acs}, we introduce Actor Communicating System and its VAS
semantics.
In Section~\ref{sec:concrete} we present a machine-based operational semantics
of \Lang. In Section~\ref{sec:cfa} we develop a general abstract interpretation
framework for \Lang\ programs, parametric on a number of basic domains. In
Section~\ref{sec:generateACS}, we use a particular instantiation of the abstract
interpretation to bootstrap the ACS construction.
In Section~\ref{sec:evaluation} we present the experimental results based on our
tool implementation \Soter, and discuss the limitations of our approach.

\paragraph{Notation.}
We write $A^*$ for the set of finite sequences of elements of the set $A$, and
$\eseq$ for the null sequence. Let $a\in A$ and $l,l'\in A^*$, we overload
`$\cons$' so that it means insertion at the top $a\cons l$, at the bottom
$l\cons a$ or concatenation $l\cons l'$.
We write $l_i$ for the $i$-th element of $l$.
The set of finite partial functions from $A$ to $B$ is denoted $A\finmap B$.
Given $f\colon A\finmap B$ we define
$f[a\mapsto b]\is
	(\lambda x.\,\textbf{if }(x\!=\!a)\textbf{ then } b \textbf{ else }f(x))$
and write $[]$ for the everywhere undefined function.

% !TEX root = main.tex
\section{A Prototypical Fragment of Erlang}
\label{sec:erlang}

In this section we introduce \emph{\Lang}, a prototypical untyped functional
language with actor concurrency. \Lang\ is essentially single-node \emph{Core
Erlang}~\cite{Carlsson:01}---the official intermediate representation of Erlang
code---without built-in functions and fault-tolerant features. It exhibits in
full the higher-order features of Erlang, with message-passing concurrency and
dynamic process creation.

\paragraph{Syntax}
\label{sec:syntax}

The syntax of \Lang\ is defined as follows:
\begin{syntax}
    e\in\dom{Exp} \prod &  x
        |	\erl{#c#(}\lstc{e}{n}\erl{)}
        |	e_0\erl{(}\lstc{e}{n}\erl{)}
        | \fun\\
        |& \erl{letrec}\ x_1 \erl{=} \fun_1. \cdots
                          x_n \erl{=} \fun_n.\
                \erl{in}\ e\\
        |&	\erl{case}\ e\ \erl{of}
                    \ \pat_1\to e_1; \ldots;
                      \pat_n\to e_n
          \ \erl{end}\\
        |&	\erl{receive}\ pat_1\to e_1; \ldots; pat_n\to e_n\ \erl{end}\\
        |&	\erl{send(} e_1, e_2 \erl{)}
        |	\erl{spawn(} e \erl{)}
        |	\erl{self()}\\
    \fun \prod& \erl{fun(}\lstc{x}{n}\erl{)} \to e\\
    \pat \prod& x \mid \erl{#c#(}\lstc{pat}{n}\erl{)}
\end{syntax}
where \erl{#c#} ranges over a finite set $\Constr$ of constructors which we
consider fixed thorough out the paper.

For ease of comparison we keep the syntax close to Core Erlang and use uncurried
functions, delimiters, \erl{fun} and \erl{end}.
We write~\mbox{`\erl{_}'} for an unnamed unbound variable; using symbols
from $\Constr$, we write $n$-tuples as
    \mbox{\erl{\{}$e_1$, \dots, $e_n$\erl{\}},}
the list constructors as cons \mbox{\erl{[\_|\_]}} and the empty list
as \mbox{\erl{[]}.} Sequencing \mbox{$(e_1$\erl{,}\ $e_2)$} is a shorthand for
\mbox{\erl{(fun(_)->}$e_2$\erl{)(}$e_1$\erl{)}} and we we omit brackets for
nullary constructors. The character `\erl{\%}' marks the
start of a line of comment. Variable names begin with an uppercase letter,
except when bound by \mbox{\erl{letrec}.}
The free variables $\freevars(e)$ of an expression are defined as usual. A
\Lang\ program \Prog\ is just a closed \Lang\ expression.

\paragraph{Labels}\label{sec:labels}
For ease of reference to program points, we associate a unique
label to each sub-expression of a program. We write $\ell\colon e$ to mean that
$\ell$ is the label associated with $e$, and we often omit the label altogether.
Take a term $\ell\colon(\ell_0\colon e_0(\ell_1\colon e_1,\ldots,\ell_n\colon
e_n))$, we define $\larg{\ell}{i}\is\ell_i$ and $\largn{\ell}\is n$.

\newcommand\onestepred{\longrightarrow}

\paragraph{Semantics}\label{sec:informal-semantics}
The semantics of \Lang\ is defined in Section~\ref{sec:semantics}, but
we informally present a small-step reduction semantics here to give an intuition of its model of concurrency.
The rewrite rules for the cases of function application and $\lambda$-abstraction are the standard ones for call-by-value $\lambda$-calculus; we write evaluation contexts as $E[\;]$.

A state of the computation of a \Lang\ program is a set $\Pi$ of processes
running in parallel. A process $\proc{e}{\pid}{\mailbox}$, identified by the pid
$\pid$, evaluates an expression $e$ with mailbox $\mailbox$ holding the messages
not yet consumed. Purely functional reductions with no side-effect take place in
each process, independently interleaved.
A \erl{spawn} construct, $ \erl{spawn(fun()->} e \erl{)} $, evaluates to a fresh
pid $\pid'$ (say), with the side-effect of the creation of a new process,
$\proc{e}{\pid'}{\emptymb}$, with pid~$\pid'$:
\[
    \proc{ E[ \erl{spawn(fun()->} e \erl{)} ] }{\pid}{\mailbox}
        \parallel \Pi
   \quad \onestepred \quad
    \proc{ E[ \pid' ] }{\pid}{\mailbox} \parallel \proc{e}{\pid'}{\emptymb}
        \parallel \Pi
\]
A \erl{send} construct, $\erl{send(} \pid, v \erl{)}$, evaluates to the message $v$ with the side-effect of appending it to the mailbox of the receiver process $\pid$; thus send is non-blocking:
\[
    \proc{ E[ \erl{send(} \pid, v \erl{)} ] }{\pid'}{\mailbox'} \parallel
    \proc{ e }{ \pid }{ \mailbox }
        \parallel \Pi
    \quad \onestepred \quad
    \proc{ E[ v ] }{\pid'}{\mailbox'} \parallel
    \proc{ e }{ \pid }{ \mailbox \cons v}
        \parallel \Pi
\]
The evaluation of a \erl{receive} construct, $ \erl{receive}\ p_1\to e_1\dots p_n\to e_n ~\erl{end}$, will block
if the mailbox of the process in question contains no message that matches any of the patterns $p_i$. Otherwise, the first message $m$ that matches a pattern, say $p_i$, is consumed by the process, and the computation continues with the evaluation of $e_i$. The pattern-matching variables in $e_i$ are bound by $\theta$ to the corresponding matching subterms of the message $m$; if more than one pattern matches the message, then (only) the first in textual order is fired.
\begin{multline*}
    \proc{ E[ \erl{receive}\ p_1\to e_1\dots p_n\to e_n~\erl{end} ] }
            {\pid}{\mailbox\cons m\cons \mailbox'}
        \parallel \Pi
    \\\onestepred
    \proc{ E[ \theta e_i ] }{\pid'}{\mailbox\cons\mailbox'}
        \parallel \Pi ;
\end{multline*}

Note that message passing is \emph{not} First-In-First-Out but rather
First-In-First-Fireable Out (FIFFO): incoming messages are queued at the end of
the mailbox but the message that matches a receive construct, and is
subsequently extracted,
is not necessarily the first in the queue.

\begin{figure}[t]
\begin{center}
\begin{erlang}[emph={lock,acquired,do,req,ans,unlock,write,read,succ,zero}]
letrec
%%% LOCKED RESOURCE MODULE
res_start= fun(Res) -> spawn(fun() -> res_free(Res)).
res_free= fun(Res) ->
    receive {lock, P} ->
        send(P, {acquired, self()}), res_locked(Res, P)
    end.
res_locked= fun(Res, P) ->
    receive
        {req, P, Cmd} ->
            case Res(P, Cmd) of
                {NewRes, ok} ->
                    res_locked(NewRes, P);
                {NewRes, {reply, A}} ->
                    send(P, {ans, self(), A}),
                    res_locked(NewRes, P)
            end;
        {unlock, P} -> res_free(Res)
    end.

% Locked Resource API
res_lock= fun(Q) -> send(Q, {lock, self()}),
                    receive {acquired, Q} -> ok end.
res_unlock= fun(Q) -> send(Q, {unlock, self()}).
res_request= fun(Q, Cmd) ->
    send(Q, {req, self(), Cmd}),
    receive {ans, Q, X} -> X end.
res_do= fun(Q, Cmd) -> send(Q, {req, self(), Cmd}).

%%% CELL IMPLEMENTATION MODULE
cell_start= fun() -> res_start(cell(zero)).
cell= fun(X) ->
        fun(_P, Cmd) ->
            case Cmd of
                {write, Y} -> {cell(Y), ok};
                 read      -> {cell(X), {reply, X}}
            end.

% Cell API
cell_lock  = fun(C) -> res_lock(C).
cell_unlock= fun(C) -> res_unlock(C).
cell_read  = fun(C) -> res_request(C, read).
cell_write = fun(C, X) -> res_do(C, {write, X}).

%%% INCREMENT CLIENT
inc = fun(C) -> cell_lock(C),
                cell_write(C, {succ, cell_read(C)}),§\label{code:criticalsec}§
                cell_unlock(C).
add_to_cell = fun(M, C) ->
    case M of   zero        ->  ok;
                {succ, M'}  ->  spawn(fun() -> inc(C)),
                                add_to_cell(M', C)
    end.
%%% ENTRY POINT
in  C = cell_start(), add_to_cell(N, C).
\end{erlang}
  \caption{Locked Resource (running example)}
  \label{fig:reslock}
\end{center}
\end{figure}

\begin{example}[Locked Resource]
\label{ex:reslock}
Figure~\ref{fig:reslock} shows an example \Lang\ program.
The code has three logical parts, which would constitute three modules
in Erlang.
The first part defines an Erlang \emph{behaviour}\footnote{I.e.~a module
implementing a general purpose protocol, parametrised over another module containing the code specific to a particular instance.} that governs the
lock-controlled, concurrent access of a shared resource by a number of clients.
A resource is viewed as a function implementing a protocol that reacts to
requests; the function is called only when the lock is acquired.
Note the use of higher-order arguments and return values.
The function \erl{res_start} creates a new process that runs an unlocked
(\erl{res_free}) instance of the resource. When unlocked, a resource
waits for a \erl{\{#lock#, P\}} message to arrive from a client \erl{P}. Upon
receipt of such a message, an acknowledgement message is sent back to the client
and the control is yielded to \erl{res_locked}. When locked (by a client
\erl{P}), a resource can accept requests \erl{\{#req#,P,Cmd\}} from
\erl{P}---and from \erl{P} only---for an unspecified command \erl{Cmd} to be
executed.

After running the requested command, the resource is expected to return the
updated resource handler and an answer, which may be the atom \erl{#ok#}, which
requires no additional action, or a couple \mbox{\erl{\{#reply#, Ans\}}} which
signals that the answer \erl{Ans} should be sent back to the client.
When an unlock message is received from \erl{P} the control is given back to
\erl{res_free}.
Note that the mailbox matching mechanism allows multiple locks and requests
to be sent asynchronously to the mailbox of the locked resource without causing
conflicts: the pattern matching in the locked state ensures that all the
pending lock requests get delayed for later consumption once the
resource gets unlocked. The functions \erl{res_lock}, \erl{res_unlock},
\erl{res_request}, \erl{res_do} encapsulate the locking protocol, hiding it from
the user who can then use this API as if it was purely functional.

The second part implements a simple `shared memory cell' resource that holds
a natural number, which is encoded using the constructors \erl{#zero#} and
\erl{\{#succ#, \_\}}, and allows a client to read its value (the command
\erl{#read#}) or overwrite it with a new one (the \erl{\{#write#, X\}} command).
Without locks, a shared resource with such a protocol easily leads to race
conditions.

The last part defines the function \erl{inc} which accesses a locked cell
to increment its value. The function \erl{add_to_cell} adds \erl{M} to the
contents of the cell by spawning \erl{M} processes incrementing it concurrently.
Finally the entry-point of the program sets up a process with a shared locked
cell and then calls \erl{add_to_cell}. Note that \erl{N} is a free variable; to
make the example a program we can either close it by setting \erl{N} to a
constant or make it range over all natural numbers with the extension described
in Section~\ref{sec:openterms}.

An interesting correctness property of this code is the mutual exclusion of the
lock-protected region (i.e. line~\ref{code:criticalsec}) of the concurrent
instances of \erl{inc}.
\end{example}

\begin{remark}\label{rem:LangVsCoreErlang}
The following Core Erlang features are not captured by \Lang.
\begin{inparaenum}
    \item Module system, exception handling, arithmetic
    primitives, built-in data types and I/O can be straightforwardly translated or integrated into our framework. They are not treated here because they are tied to the inner workings of the Erlang runtime system.
    \item Timeouts in receives, registered processes and type guards can be supported using suitable abstractions.
     \item A proper treatment of monitor / link primitives and the multi-node semantics will require a major extension of the concrete (and abstract) semantics.
\end{inparaenum}
\end{remark}

% !TEX root = main.tex
\section{Actor Communicating Systems}
\label{sec:acs}

In this section we explore the design space of abstract models of Erlang-style
concurrency. We seek a model of computation that should capture the core
concurrency and asynchronous communication features of \Lang\ and yet enjoys the
decidability of interesting verification problems.  In the presence of
pattern-matching algebraic data types, the (sequential) functional fragment of
\Lang\ is already Turing powerful \cite{Ong:11}. Restricting it to a pushdown
(equivalently, first-order) fragment but allowing concurrent execution would
enable, using very primitive synchronization, the simulation of a
Turing-powerful finite automaton with two stacks.
A single finite-control process equipped with a mailbox (required for
asynchronous communication) can encode a Turing-powerful queue automaton in the
sense of Minsky. Thus constrained, we opt for a model of concurrent computation
that has finite control, a finite number of messages, and a finite number of
\emph{process classes}.

\begin{definition}\label{def:acs}
    An \emph{Actor Communicating System} (ACS) $\Acs$ is a tuple
    $
        \tuple{P, Q, M, R, \pid_0, q_0}
    $
    where $P$ is a finite set of \emph{pid-classes}, $Q$ is a finite set of
    control-states, $M$ is a finite set of messages, $\pid_0\in P$ is the
    pid-class of the initial process, $q_0\in Q$ is the initial state of the
    initial process, and $R$ is a finite set of rules of the form
    $\ecsrule[\pid]{q}{\ell}{q'}$ where $\pid\in P$, $q,q'\in Q$ and $\ell$
    is a label that can take one of four possible forms:
    \begin{itemize}
	    \item $\tau$, which represents an internal (sequential) transition of
              a process of pid-class $\pid$
	    \item $?m$ with $m\in M$: a process of pid-class $\pid$ extracts
              (and reads) a message $m$ from its mailbox
	    \item $\pid'!m$ with $\pid'\in P$, $m\in M$: a process of pid-class
              $\pid$ sends a message $m$ to a process of pid-class $\pid'$
	    \item $\nu \pid'.\: q''$ with $\pid'\in P$ and $q''\in Q$: a process of
              pid-class $\pid$ spawns a new process of pid-class $\pid'$ that
              starts executing from $q''$
    \end{itemize}
\end{definition}

Now we have to give ACS a semantics, but interpreting the ACS mailboxes as FIFFO queues would yield a Turing-powerful model.
Our solution is to apply a \emph{counter abstraction} on mailboxes: disregard
the ordering of messages, but track the number of occurrences of every message
in a mailbox. Since we bound the number of pid-classes, but wish to model
dynamic (and hence unbounded) spawning of processes, we apply a second counter
abstraction on the control states of each pid-class: we count, for each
control-state of each pid-class, the number of processes in that pid-class that
are currently in that state.

It is important to make sure that such an abstraction contains all the
behaviours of the semantics that uses FIFFO mailboxes: if there is a term in the
mailbox that matches a pattern, then the corresponding branch is
non-deterministically fired. To see the difference, take the ACS that has one
process (named~$\pid$), three control states $q$, $q_1$ and $q_2$, and two rules
$\eRec[q--?a-->q_1]$, $\eRec[q--?b-->q_2]$. When equipped with a FIFFO mailbox
containing the sequence $c\,a\,b$, the process can only evolve from $q$ to $q_1$
by consuming $a$ from the mailbox, since it can skip $c$ but will find a
matching message (and thus not look further into the mailbox) before reaching
the message $b$. In contrast, the VAS semantics would let $q$ evolve
non-deterministically to both $q_1$ and $q_2$, consuming $a$ or $b$
respectively: the mailbox is abstracted to $[a\mapsto 1, b\mapsto 1, c\mapsto
1]$ with no information on whether $a$ or $b$ arrived first.
However, the abstracted semantics does contain the traces of the FIFFO
semantics.

The VAS semantics of an ACS is a state transition system equipped with
counters (with values in $\Nat$) that support increment and decrement (when
non-zero) operations. Such infinite-state systems are known as \emph{vector
addition systems} (VAS), which are equivalent to Petri~nets.

\newcommand{\VasRules}{\mathbf{R}}
\begin{definition}[Vector Addition System]\label{def:vas}
    \begin{asparaenum}
	\item A \emph{vector addition system} (VAS) $\VAS$ is a pair $(I,R)$ where
		  $I$ is a finite set of indices (called the \emph{places} of the VAS)
		  and $R\subseteq\Int^I$ is a finite set of \emph{rules}.
		  Thus a rule is just a vector of integers of dimension $|I|$,
	      whose components are indexed (i.e.~named) by the elements of $I$.
	\item The \emph{state transition system} $\sem{\VAS}$ induced by a VAS $\VAS=(I,R)$ has state-set $\Nat^I$ and transition relation
	      \[\makeset{(\vec{v}, \vec{v}+\vec{r}) \mid
	      		\vec{v} \in \Nat^{I}, \vec{r} \in R, \vec{v} + \vec{r} \in \Nat^{I}}.\]
    \end{asparaenum}
	\noindent We write $\vec{v}\leq\vec{v}'$ just if
	for all $i$ in $I$, $\vec{v}(i)\leq\vec{v}'(i)$.
\end{definition}

The semantics of an ACS can now be given easily in terms of a corresponding
underlying vector addition system:

\begin{definition}[VAS semantics]
    \label{def:acs2vas}\label{def:parikhsem}
    The semantics of an ACS $\Acs=(P, Q, M, R, \pid_0, q_0)$
    is the transition system induced by the VAS $\VAS = (I, \VasRules)$ where
    $I = P \times (Q \dunion M)$ and $\VasRules = \makeset{\vec{r} \mid r \in R})$.
    The transformation $r \mapsto \vec{r}$ is defined as follows.
    \footnote{All unspecified components of the vectors $\vec{r}$ as defined
    in the table are set to zero.}
    \begin{center}
    \begin{tabular}{ll}
    	\toprule
    	\multicolumn{1}{c}{{\bf ACS Rules}: $r$} &
    	\multicolumn{1}{c}{{\bf VAS Rules}: $\vec{r}$}\\
    	\midrule
    	$\eTau\pid[q --> q']$ &
    		$[(\pid,q)\mapsto -1, (\pid,q')\mapsto 1]$\\
    	$\eRec\pid[q --?m--> q']$ &
    		$[(\pid,q)\mapsto -1, (\pid,q')\mapsto 1,(\pid,m)\mapsto -1]$\\
    	$\eSend\pid[q --\pid'!m--> q']$ &
    		$[(\pid,q)\mapsto -1, (\pid,q')\mapsto 1, (\pid',m)\mapsto 1]$\\
    	$\eSpawn\pid[q--v \pid'.\: q''--> q']$ &
    		$[(\pid,q)\mapsto -1, (\pid,q')\mapsto 1, (\pid',q'')\mapsto 1]$\\
    	\bottomrule
    \end{tabular}
    \end{center}
    Given a $\sem{\VAS}$-state $\vec{v} \in \Nat^{I}$, the component
    $\vec{v}(\pid,q)$ counts the number of processes in the pid-class $\pid$
    currently in state $q$, while the component $\vec{v}(\pid,m)$ is the sum of the
    number of occurrences of the message $m$ in the
    mailboxes of the processes of the pid-class $\pid$.
\end{definition}

While infinite-state, many non-trivial properties are decidable on VAS including reachability, coverability and place boundedness; for more details see~\cite{Finkel:01}.
In this paper we focus on coverability, which is \expspace-complete~\cite{Rackoff:78}: given two states $s$ and $t$, is it possible to reach
from $s$ a state $t'$ that covers $t$ (i.e.~$t'\leq t$)?

\medskip

Which kinds of correctness properties of \Lang\ programs can one specify by
coverability of an ACS?
We will be using ACS to \emph{over-approximate} the semantics of a \Lang\
program, so if a state of the ACS is not coverable, then it is not reachable in
any execution of the program.
It follows that we can use coverability to express safety properties such as:
\begin{inparaenum}
  \item unreachability of error program locations
  \item mutual exclusion
  \item boundedness of mailboxes: is it possible to reach a state where the mailbox of pid-class $\pid$ has more than $k$ messages? If not we can allocate just $k$ memory cells for that mailbox.
\end{inparaenum}

% !TEX root = main.tex
\section{An Operational Semantics for \Lang}
\label{sec:concrete}
In this section, we define an operational semantics for \Lang\ using a
\emph{time-stamped CESK* machine}, following a methodology advocated by Van Horn
and Might \cite{VanHorn:10}.
An unusual feature of such machines are \emph{store-allocated continuations}
which allow the
recursion in a programs's control flow and data structure to be separated from the recursive structure in its state space.
As we shall illustrate in Section~\ref{sec:cfa}, such a formalism is key to a \emph{transparently sound} and \emph{parametric} abstract interpretation.

\paragraph{A Concrete Machine Semantics.}
\label{sec:cesk}\label{sec:semantics}
Without loss of generality, we assume that in a \Lang\ program, variables are
distinct, and constructors and cases are only applied to (bound) variables. The \Lang\ machine defines a transition system on \emph{(global) states},
which are elements of the set \State
\begin{align*}
	s\in\State       & \is \Procs \times \Mailboxes \times \Store\\
	\pi\in\Procs     & \is \Pid\finmap\ProcState\\
	\mu\in\Mailboxes & \is \Pid\finmap\Mailbox
\end{align*}
An element of \Procs\ associates a process with its \emph{(local) state}, and an element of \Mailboxes\ associates a process with its mailbox.
We split the $\Store$ into two partitions
\[
    \sigma\in\Store \is (\AddrVar \finmap \Val) \times (\AddrKont \finmap \Kont)
\]
each with its address space, to separate \emph{values}
and \emph{continuations}. By abuse of notation $\sigma(x)$ shall mean the application of the first component when $x\in\AddrVar$ and of the second when $x\in\AddrKont$.

The \emph{local state} of a process
\[
	q\in\ProcState \is
        (\ProgLoc\dunion\Pid) \times \Env \times \AddrKont \times \Time\\
\]
is a tuple, consisting of
\begin{inparaenum}
    \item a pid, or a \emph{program location}%
          \footnote{Precisely a program location is a node in the abstract
                    syntax tree of the program being analysed.}
          which is a subterm of the program, labelled with its occurrence;
          whenever it is clear from the context, we shall omit the label;

    \item an environment, which is a map from variables to pointers to values $\rho \in \Env \is \Var \finmap \AddrVar$;

    \item a pointer to a continuation, which indicates what to evaluate next
          when the current evaluation returns a value;

    \item a time-stamp, which will be described later.
\end{inparaenum}

\emph{Values} are either closures or pids:
\[
	d \in \Val \is \Closure \dunion \Pid \qquad
		  \Closure \is \ProgLoc \times \Env
\]
Note that, as defined, closures include both functions (which is standard) as
well as constructor terms.

All the domains we define are naturally partially ordered:
$\ProgLoc$ and $\Var$ are discrete partial orders, all the others are defined by the appropriate pointwise extensions.

\paragraph{Mailbox and Message Passing}

A \emph{mailbox} is just a finite sequence of values:
$\mailbox\in\Mailbox \is \Val^*$.
We denote the empty mailbox by $\emptymb$.
A mailbox is supported by two operations:
\begin{align*}
	\mailmatch
        & \colon \pat^\ast \times \Mailbox \times \Env \times \Store\to    \\
		& \quad\qquad {(\Nat\times(\Var\finmap\Val)\times\Mailbox)}_{\bot} \\
	\enqueue &\colon \Val \times \Mailbox \to \Mailbox
\end{align*}
The function $\mailmatch$ takes a list of patterns, a mailbox, the current
environment and a store (for resolving pointers in the values stored in the
mailbox) and returns the index of the matching pattern, a substitution
witnessing the match, and the mailbox resulting from the extraction of the
matched message. To model \emph{Erlang-style} FIFFO mailboxes we set $\enqueue(d, \mailbox) := \mailbox\cons d$ and define:
\[	\mailmatch(\lst{p}{n},\mailbox,\rho,\sigma)\is
		(i,\theta,\mailbox_1\cons\mailbox_2)\]
such that
\[\begin{array}{lr}
		\mailbox = \mailbox_1\cons d\cons \mailbox_2
		&
		\forall d'\in \mailbox_1 \, .  \,\forall j \, .\,
			{\undef{\match_{\rho,\sigma}}{p_j, d'}}\\
		\theta = \match_{\rho,\sigma}(p_i, d)
		&
		\forall j < i \, . \,
			{\undef{\match_{\rho,\sigma}}{p_j, d}}
\end{array}\]
where $\match_{\rho,\sigma}(p, d)$ seeks to match the term $d$ against the pattern $p$,  following the pointers $\rho$ to the store $\sigma$ if necessary, and returning the witnessing substitution if matchable, and $\bot$ otherwise.

\paragraph{Evaluation Contexts as Continuations.}

Next we represent (in an inside-out manner) evaluation contexts as
continuations. A continuation consists of a \emph{tag} indicating the shape of
the evaluation context, a pointer to a continuation representing the enclosing
evaluation context, and, in some cases, a program location and an environment.
Thus $\kont\in\Kont$ consists of the following constructs:
\begin{asparaitem}
  \item $\kStop$ represents the empty context.
  \item $\kArg{i}{\ell, v_{0}\ldots v_{i-1}, \rho, a}$
        represents the context
        \[ E[v_0(\lstc{v}{i-1},[\,],\lstc[i+1]{e'}{n})] \]
        where
        $e_0(\lstc{e}{n})$ is the subterm located at $\ell$;
        $\rho$ closes the terms $\lstc[i+1]{e}{n}$ to  $\lstc[i+1]{e'}{n}$ respectively;
        the address $a$ points to the continuation representing the enclosing
        evaluation context $E$.
\end{asparaitem}

\paragraph{Addresses, Pids and Time-Stamps.}

While the machine supports arbitrary concrete representations of time-stamps,
addresses and pids, we present here an instance based on
\emph{contours}~\cite{Shivers:91} which shall serve as the reference semantics
of \Lang, and the basis for the abstraction of Section~\ref{sec:cfa}.

A way to represent a \emph{dynamic} occurrence of a symbol is the history of the
computation at the point of its creation. We record history as \emph{contours}
which are strings of program locations
\[
    \cntr\in\Time \is \ProgLoc^*
\]
The initial contour is just the empty sequence $\cntr_0\is\eseq$, while the
$\tick$ function updates the contour of the process in question by prepending
the current program location, which is always a function call
(see rule \ref{rule:apply}):
\[
    \tick \colon \ProgLoc \times \Time \to \Time\qquad
    \tick (\ell,\cntr) \is \ell \cons \cntr
\]

Addresses for values ($b\in\AddrVar$) are represented by tuples comprising the current pid, the variable in question, the bound value and the current time stamp.
Addresses for continuations ($a, c\in\AddrKont$) are represented by tuples comprising the current pid, program location, environment and time (i.e.~contour); or $\StopAddr$ which is the address of the initial continuation ($\kTop$).
\begin{align*}
    \AddrVar  &\is \Pid\times\Var\times\Data\times\Contour\\
    \AddrKont &\is (\Pid\times\ProgLoc\times\Env\times\Contour)
                    \dunion\{\StopAddr\}
\end{align*}
The \emph{data domain} ($\data\in\Data$) is the set of closed \Lang\ terms; the function $\resolve\colon \Store\times\Val\to\Data$ resolves all the pointers of a value through the store $\sigma$, returning the corresponding closed term:
\begin{align*}
    \resolve(\sigma, \pid)&\is \pid\\
    \resolve(\sigma, (e,\rho)) &
        \is e[x\mapsto \resolve(\sigma, \sigma(\rho(x)))\mid x\in\freevars(e)]
\end{align*}

New addresses are allocated by extracting the relevant components from the context at that point:
\begin{align*}
    \newKPush & \colon \Pid \times \ProcState \to \AddrKont\\
    \newKPush & (\pid,\tuple{\ell,\rho,\_, \cntr})
        \is (\pid,\larg{\ell}{0},\rho,\cntr)\\%[.5em]
    \newKPop  & \colon \Pid \times \Kont \times \ProcState \to \AddrKont\\
    \newKPop & (\pid,\kont,\tuple{\_,\_,\_,\cntr})
        \is (\pid,\larg{\ell}{i+1},\rho,\cntr)\\
            & \text{ where }
                \kont=\kArg{i}{\ell,\dots,\rho,\_}\\%[.5em]
    \newVaddr & \colon \Pid \times \Var \times \Data \times \ProcState
        \to \AddrVar\\
    \newVaddr & (\pid,x,\data,\tuple{\_,\_,\_,\cntr})
        \is (\pid,x,\data,\cntr)
\end{align*}

\begin{remark}
    To enable data abstraction in our framework, the address of a
    value contains the data to which the variable is bound: by making
    appropriate use of the embedded information in the abstract semantics, we
    can fine-tune the data-sensitivity of our analysis, as we shall illustrate
    in Section~\ref{sec:kCFA}. However when no data abstraction is intended,
    this data component can safely be discarded.
\end{remark}

Following the same scheme, pids ($\pid\in\Pid$) can be identified with the
contour of the \erl{spawn} that generated them:
$\Pid \is (\ProgLoc\times\Contour)$.
Thus the generation of a new pid is defined as
\begin{align*}
    \newPid & \colon \Pid \times \ProgLoc \times \Time \to \Pid\\
    \newPid & ((\ell',\cntr'),\ell,\cntr) \is
        (\ell, \tick^*(\cntr, \tick(\ell',\cntr'))
\end{align*}
where $\tick^*$ is just the simple extension of $\tick$ that prepends a whole
sequence to another. Note that the new pid contains the pid that created it as a
sub-sequence: it is indeed part of its history (dynamic context). The pid
$\pid_0\is(\ell_0,\emptycntr)$ is the pid associated with the starting process,
where $\ell_0$ is just the root of the program.

\begin{remark}\label{rem:findomains}
\begin{inparaenum}

\item
    Note that the only sources of infinity for the state space are time,
    mailboxes and the data component of value addresses. If these domains are
    finite then the state space is finite and hence reachability is  decidable.

\item
    It is possible to present a more general version of the concrete machine
    semantics. We can reorganise the machine semantics so that components such
    as \Time, \Pid, \Mailbox, \AddrKont\ and \AddrVar\ are presented as
    \emph{parameters} (which may be instantiated as the situation requires).
    In this paper we present a contour-based machine, which is general enough to illustrate our method of verification.

\end{inparaenum}
\end{remark}

\begin{definition}[Concrete Semantics]
Now that the state space is set up, we define a (non-deterministic) transition relation on states $(\semTo)\subseteq\State\times\State$.
In~Figure~\ref{fig:concrete-rules} we present the rules for application,
message passing and process creation; we omit the other rules (letrec, case and treatment of pids as returned value) since they follow the same shape. The transition $s\semTo s'$ is defined by a case analysis of the shape of $s$.
\end{definition}

%%%%%%%%%%%%%%%%%%%%%%%% RULES OF THE CONCRETE MACHINE %%%%%%%%%%%%%%%%%%%%%%%%%

\begingroup

\def\arraystretch{1}
\def\smallsep{}%\midrule[.02em]}
\def\midspace{\hspace{2mm}}

\begin{figure*}[t]
\noindent
\begingroup
\centering%
\small
\setlength{\colwd}{6.25cm}
\begin{tabular}[t]{@{}l@{}}
    \toprule
    \RulesSection{Functional reductions}\\
    \midrule

    \RuleName{FunEval}\label{rule:funeval}\\
    \CondCell{
        \pi(\pid) =  & \tuple{\ell\colon(e_0(\lstc{e}{n})), \rho, a, \cntr}\\
               b \is & \newKPush(\pid, \pi(\pid))
    }\\\smallsep
    \ResCell{
           \pi' &= \pi[\pid\mapsto\tuple{e_0,\rho,b,\cntr}]\\
        \sigma' &= \sigma[b\mapsto\kArg{0}{\ell, \eseq, \rho,a}]
    }\\\midrule

    \RuleName{ArgEval}\label{rule:argeval}\\
    \CondCell{
        \pi(\pid) =  & \tuple{v, \rho, a, \cntr}\\
        \sigma(a) =  & \kont = \kArg{i}{\ell,\lst[0]{d}{i-1},\rho',c}\\
             d_i \is & (v,\rho)\\
               b \is & \newKPop(\pid, \kont, \pi(\pid))
    }\\\smallsep
    \ResCell{
           \pi' &= \pi[\pid\mapsto\tuple{\larg{\ell}{i+1},\rho',b,\cntr}]\\
        \sigma' &= \sigma[b\mapsto\kArg{i+1}{\ell,\lst[0]{d}{i},\rho',c}]
    }\\\midrule

    \RuleName{Apply}\label{rule:apply}\\
    \CondCell{
        \pi(\pid) =  & \tuple{v,\rho,a,\cntr}, \ \largn{\ell} = n\\
        \sigma(a) =  & \kont = \kArg{n}{\ell,\lst[0]{d}{n-1},\rho',c}\\
              d_0 =  & (\erl{fun(}\lst{x}{n}\erl{)}\to e, \rho_0) \quad
             d_n \is   (v, \rho) \\
             b_i \is & \newVaddr(\pid, x_i, \resolve(\sigma, d_i), \pi(\pid))\\
          \cntr' \is & \tick(\ell, \pi(\pid))
    }\\\smallsep
    \ResCell{
        \pi' &= \pi[\pid\mapsto
                        \tuple{e,\rho'[x_1\to b_1\ldots x_n\to b_n],c,\cntr'}]\\
     \sigma' &= \sigma[b_1\mapsto d_1\ldots b_n\mapsto d_n]\\
    }\\\bottomrule
\end{tabular}
\midspace%
%\begin{minipage}[t]{6.6cm}
\setlength{\colwd}{6.4cm}
\begin{tabular}[t]{@{}l@{}}
    \toprule
    \RuleName{Vars}\label{rule:vars}\\[4pt]
    \CondCell{
        &      \pi(\pid) = \tuple{x,\rho,a,\cntr}\\
        &\sigma(\rho(x)) = (v,\rho')
    }\\\smallsep
    \ResCell{
        \pi'&=\pi[\pid\mapsto\tuple{v,\rho',a,\cntr}]
    }\\[5pt]
    \midrule[\heavyrulewidth]
    \RulesSection{Communication}\\
    \midrule
    \RuleName{Receive}\label{rule:receive}\\
    \CondCell{%[when\\\null\quad]{
        & \pi(\pid) =\tuple{
            \erl{receive}\ p_1 \to e_1\ldots p_n\to e_n\ \erl{end},
            \rho,a,\cntr}\\
        & \mailmatch(\lst{p}{n},\mu(\pid),\rho,\sigma) = (i,\theta,\mailbox)\\
        & \theta = [x_1\mapsto d_1\ldots x_k\mapsto d_k]\\
        & b_j \is \newVaddr(\pid, x_j, \resolve(\sigma, d_j), \pi(\pid))\\
        &\rho'\is \rho\;\![x_1\mapsto b_1\ldots x_k\mapsto b_k]
    }\\\smallsep
    \ResCell{
           \pi'&= \pi[\pid\mapsto\tuple{e_i, \rho',a,\cntr}]\\
           \mu'&= \mu[\pid\mapsto \mailbox]\\
        \sigma'&=\sigma[b_1\mapsto d_1\ldots b_k\mapsto d_k]
    }\\
    \midrule
    \RuleName{Send}\label{rule:dispatch}\label{rule:send}\\
    \CondCell{
        \pi(\pid) =& \tuple{v,\rho,a,\cntr}\\
        \sigma(a) =& \kont = \kArg{2}{\ell,d,\iota',\_,c}\\
                d =& (\erl{send}, \_ )
    }\\\smallsep
    \ResCell{
        \pi' &= \pi[\pid\mapsto\tuple{v, \rho, c,\cntr}]\\
        \mu' &= \mu[\pid'\mapsto\enqueue((v,\rho),\mu(\pid'))]%
    }\\
    \bottomrule
\end{tabular}
%\end{minipage}
\midspace%
\setlength{\colwd}{4.25cm}
\begin{tabular}[t]{@{}l@{}}
    \toprule
    \RulesSection{Process creation}\\
    \midrule
    \RuleName{Spawn}\label{rule:proccr}\label{rule:spawn}\\
    \CondCell{ %[when\\\null\quad]{
        \pi(\pid) =& \tuple{\erl{fun()}\to e,\rho,a,\cntr}\\
        \sigma(a) =& \kArg{1}{\ell, d, \rho',c}\\
        d=& (\erl{spawn},\_)\\
        \pid'    \is & \newPid(\pid, \ell, \timest)
        % v =& \erl{fun()}\to e\\
        % b_i \is& \newVaddr(\pid, x_i, d_i, \sigma, \timest)
    }\\\smallsep
    \ResCell[then\\\null\quad]{
        \pi' &= \pi\left[
        \begin{aligned}
            \pid  & \mapsto\tuple{\pid', \rho', c,\cntr},\\
            \pid' & \mapsto\tuple{e,\rho,\StopAddr,\cntr_0}
        \end{aligned}\right]\\[10pt]
%       \sigma' &= \sigma[b\mapsto\kTop]\\
        \mu' &= \mu[\pid'\mapsto\emptymb]
    }\\\midrule

    \RuleName{Self}\label{rule:self}\\
    \CondCell{
        \pi(\pid)=\tuple{\erl{self()},\rho,a,\cntr}
    }\\\smallsep
    \ResCell{
        \pi'=\pi[\pid\mapsto\tuple{\pid,\rho,a,\cntr}]
    }\\
    \midrule[\heavyrulewidth]
    \RulesSection{Initial state}\\
    \midrule
    \parbox[t]{\colwd}{\rule[9pt]{0pt}{0pt}%
    \RuleName{Init}\label{rule:init}
        The initial state associated~with a program $\Prog$ is
        $s_\Prog\is\tuple{\pi_0,\mu_0,\sigma_0}$
        where
        $\begin{array}[t]{@{}r@{\;}l@{}}
            \pi_0 &= [\pid_0\mapsto \tuple{\Prog,[],\StopAddr,\cntr_0}]\\
            \mu_0 &= [\pid_0\mapsto\emptymb]\\
            \sigma_0 &= [\StopAddr\mapsto\kTop]
        \end{array}$\\[2pt]
    }\\
    \bottomrule
\end{tabular}
\endgroup

\nocaptionrule
\caption{\label{fig:concreteSemantics}
    \textbf{Operational Semantics Rules}.
    The tables define the transition relation
    $s=\tuple{\pi,\mu,\sigma,\timest} \semTo\tuple{\pi',\mu',\sigma',\timest'}=s'$ by cases;
    the primed components of the state are identical to
    the non-primed components, unless indicated otherwise in the ``then'' part of the rule.
    The meta-variable $v$ stands for terms that cannot be further rewritten such
    as $\lambda$-abstractions, constructor applications and un-applied primitives.
}
\label{fig:concrete-rules}

\end{figure*}
\endgroup

%%%%%%%%%%%%%%%%%%%%%%%%%%%%%%%%%%%%%%%%%%%%%%%%%%%%%%%%%%%%%%%%%%%%%%%%%%%%%%%%

The rules for the purely functional reductions are a simple lifting of the
corresponding rules for the sequential CESK* machine: when the currently
selected process is evaluating a variable~\ref{rule:vars} its address is looked
up in the environment and the corresponding value is fetched from the store and
returned. \ref{rule:apply}: When evaluating an application, control is given to
each argument---including the function to be applied---in turn;
\ref{rule:funeval} and \ref{rule:argeval} are then applied, collecting the
values in the continuation. After all arguments have been evaluated, new values
are recorded in the environment (and the store), and control is given to the
body of the function to be applied. The rule \ref{rule:receive} can only fire if
$\mailmatch$ returns a valid match from the mailbox of the process. In case
there is a match, control is passed to the expression in the matching clause,
and the substitution $\theta$ witnessing the match is used to generate the
bindings for the variables of the pattern. When applying a send \ref{rule:send},
the recipient's pid is first extracted from the continuation, and $\enqueue{}$
is then called to dispatch the evaluated message to the designated mailbox. When
applying a spawn \ref{rule:spawn}, the argument must be an evaluated nullary
function; a new process with a fresh pid is then created whose code is the body
of the function.

One can easily add rules for run-time errors such as wrong arity in function
application, non-exhaustive patterns in cases, sending to a non-pid and spawning
a non-function.

% !TEX root = main.tex
\section{Parametric Abstract Interpretation}
\label{sec:cfa}
\label{sec:absmachine}

We aim to abstract the concrete operational semantics of Section~\ref{sec:concrete} isolating the least set of domains that need to be made finite in order for the abstraction to be decidable. We then state the conditions on these abstract domains that are sufficient for soundness.

In Remark~\ref{rem:findomains} we identify $\Time$, $\Mailbox$ and $\Data$ as responsible for the unboundedness of the state space. Our abstract semantics is thus parametric on the abstraction of these basic domains.

\begin{definition}[Basic domains abstraction]\label{def:bda}
\begin{asparaenum}
    \item A \emph{data abstraction} is a triple
        $\DataAbs = \tuple{
            \abs\Data,
            \absf[d],
            \abs\resolve
        }$
        where $\abs\Data$ is a flat {(i.e.~discretely ordered)} domain of abstract data values,
        $\absf[d]\colon \Data\to\abs\Data$ and
        $\abs\resolve \colon
            \abs\Store\times\abs\Value \to \powerset(\abs\Data)$.

    \item A \emph{time abstraction} is a tuple
        $\TimeAbs = \tuple{
            \abs\Time,
            \absf[t],
            \abs\tick,
            \abs\cntr_0
        }$
        where $\abs\Time$ is a flat domain of abstract contours,
        $\absf[t] \colon \Time \to \abs\Time$, $\abs\cntr_0\in\abs\Time$, and
        $\abs\tick \colon \ProgLoc \times \abs\Time \to \abs\Time$.

    \item A \emph{mailbox abstraction} is a tuple
        $\MailboxAbs = \tuple{
            \abs\Mailbox,
            \leqmb,
            \join_{\text{m}},
            \absf[m],
            \abs\enqueue,
            \abs\emptymb,
            \abs\mailmatch
        }$
        where $(\abs\Mailbox, \leqmb, \join_{\text{m}})$ is a join-semilattice
        with least element $\abs\emptymb\in\abs\Mailbox$,
        $\absf[m] \colon \Mailbox \to \abs\Mailbox$ and
        $\abs\enqueue\colon\abs\Val \times \abs\Mailbox \to \abs\Mailbox$
        are monotone in mailboxes.
        \begin{multline*}
            \abs\mailmatch \colon
                \pat^\ast \times \abs\Mailbox \times \abs\Env \times \abs\Store
                    \to\\
                \powerset(\Nat\times(\Var\finmap\abs\Val)\times\abs\Mailbox)
        \end{multline*}

    \item A \emph{basic domains abstraction} is a triple
        $\AbsInt = \tuple{\DataAbs, \TimeAbs, \MailboxAbs}$
        consisting of a data, a time and a mailbox abstraction.
\end{asparaenum}
\end{definition}

An abstract interpretation of the basic domains determines an interpretation of the other abstract domains as follows.
\begin{align*}
    \abs\State
        & \is \abs\Procs \times \abs\Mailboxes \times \abs\Store \\
    \abs\Procs
        & \is \abs\Pid\to\powerset(\abs\ProcState) \\
    \abs\ProcState
        & \is (\ProgLoc\dunion\abs\Pid) \times
              \abs\Env \times \abs\AddrKont \times
              \abs\Time \\
    \abs\Store
        & \is (\!\abs\AddrVar  \to \powerset(\abs\Val))
                  \times
              (\abs\AddrKont \to \powerset(\abs\Kont)) \\
    \abs\Mailboxes &\is \abs\Pid\to\abs\Mailbox \qquad
    \abs\Val \is \abs\Closure \dunion \abs\Pid\\
    \abs\Closure
        &\is \ProgLoc \times \abs\Env\qquad
    \abs\Env
         \is \Var \finmap \abs\AddrVar\\
    \abs\Pid
         &\is (\ProgLoc\times\abs\Time)\dunion\{\abs\pid_0\} \qquad
    \abs\pid_0\is \abs\cntr_0
\end{align*}
each equipped with an abstraction function defined by an appropriate pointwise extension. We will call all of them $\absf$ since it will not introduce ambiguities. The abstract domain $\abs\Kont$ is the pointwise abstraction of $\Kont$, and we will use the same tags as those in the concrete domain.
The abstract functions $\abs\newKPush$, $\abs\newKPop$, $\abs\newVaddr$ and $\abs\newPid$, are defined exactly as their concrete versions, but on the abstract domains.

When $B$ is a flat domain, the abstraction of a partial map $C = A\finmap B$ to $\abs C = \abs A\to\powerset(\abs B)$ \label{sec:absmachine:partialmap}
is defined as
\begin{equation*}\alpha_C(f)\is
    \lambda \abs a\in \abs A.
        \ \{\alpha_B(b) \mid (a,b)\in f \text{ and } \alpha_A(a)=\abs a\}
\end{equation*}
where the preorder on $\abs C$ is
$\abs f\leq_{\abs C} \abs g \iif
    \forall \abs a.\ \abs f(\abs a) \subseteq g(\abs a)$.

The operations on the parameter domains need to `behave' with respect to the
abstraction functions: the standard correctness conditions listed below must be
satisfied by their instances. These conditions amount to requiring that what we
get from an application of a concrete auxiliary function is adequately represented
by the abstract result of the application of the abstract counterpart of that
auxiliary function. The partial orders on the domains are standard pointwise
extensions of partial orders of the parameter domains.

\begin{definition}[Sound basic domains abstraction]\label{def:soundAbstraction}
    A basic domains abstraction $\AbsInt$ is \emph{sound}
    just if the following conditions are met by the auxiliary operations:
    \begin{gather}
        \absf[t](\tick(\ell,\cntr)) \leq \abs\tick(\ell, \absf[t](\cntr))
        \label{eq:monotonetick}\\
        \abs\sigma \leq \abs\sigma' \wedge \abs d\leq \abs d'
            \implies
        \absresolve(\abs\sigma, \abs d)\leq\absresolve(\abs\sigma', \abs d')
        \label{eq:monotoneres}\\
        \forall \abs\sigma \geq \absf(\sigma).\;
            \absf[d](\resolve(\sigma, d))\in \absresolve(\abs\sigma, \absf(d))
        \label{eq:soundres}\\
        \absf[m](\enqueue(d, \mailbox)) \leq
            \abs\enqueue(\absf(d), \absf[m](\mailbox)) \quad
        \absf[m](\emptymb)=\abs\emptymb
        \label{eq:soundenq}\\
    \intertext{
    if $\mailmatch(\vect{p},\mailbox,\rho,\sigma)=(i,\theta,\mailbox')$ then
    {$\forall \abs\mailbox \geq \absf(\mailbox)$, $\forall \abs\sigma \geq \absf(\sigma)$,} $\exists \abs\mailbox'\geq\absf(\mailbox')$
    such that}
    (i,\absf(\theta),\abs\mailbox') \in
            \abs\mailmatch(\vect{p},{\abs\mailbox},\absf(\rho),{\abs\sigma})
    \label{eq:soundmmatch}
    \end{gather}

\end{definition}
Following the Abstract Interpretation framework, one can exploit the soundness
constraints to derive, by algebraic manipulation, the definitions of the
abstract auxiliary functions which would then be correct by
construction~\cite{Midtgaard:08}.

\begin{definition}[Abstract Semantics]
    Once the abstract domains are fixed, the rules that define the abstract
    transition relation are straightforward abstractions of the original ones.
    In Figure~\ref{fig:abstract-rules}, we present the abstract counterparts of
    the rules for the operational semantics in Figure~\ref{fig:concrete-rules},
    defining the non-deterministic abstract transition relation on abstract
    states
    $ (\cfaTo)\subseteq\abs\State\times\abs\State. $
    When referring to a particular program $\Prog$, the abstract semantics is
    the portion of the graph reachable from $s_{\Prog}$.
\end{definition}

%%%%%%%%%%%%%%%%%%%%%%% RULES OF THE ABSTRACT MACHINE %%%%%%%%%%%%%%%%%%%%%%%%%

\begingroup

\def\arraystretch{1}
\def\smallsep{}%\midrule[.02em]}
\def\midspace{\hspace{2mm}}

\begin{figure*}[t]
\noindent
\begingroup
\centering%
\small
\setlength{\colwd}{6.3cm}
\begin{tabular}[t]{@{}l@{}}
    \toprule
    \RulesSection{Functional abstract reductions}\\
    \midrule

    \AbsRuleName{FunEval}\label{absrule:funeval}\\
    \CondCell{
        \abs\pi(\abs\pid)&\ni\abs q = \tuple{\ell\colon(e_0(\lstc{e}{n})),\abs\rho,\abs a,\abs t\:}\\
        \abs b \is& \abs\newKPush(\abs\pid,\abs q)
        % \abs b \is& \abs\newKaddr(\abs\pid,\larg{\ell}{0},\abs\timest)
    }\\\smallsep
    \ResCell{
        \abs\pi' &= \abs\pi\join
            [\abs\pid\mapsto\makeset{\tuple{e_0,\abs \rho,\abs b,\abs t\:}}]\\
        \abs\sigma' &= \abs\sigma\join
            [\abs b\mapsto\makeset{\kArg{0}{\ell, \eseq, \abs\rho,\abs a}}]
    }\\\midrule

    \AbsRuleName{ArgEval}\label{absrule:argeval}\\
    \CondCell{
        \abs\pi(\abs\pid)\ni&\tuple{v,\abs\rho,\abs a,\abs t\:}\\
        \abs\sigma(\abs a)\ni& \abs\kont = \kArg{i}{\ell,\lst[0]{\abs d}{i-1},\abs\rho',\abs c}\\
        \abs d_i\is& (v,\abs\rho)\\
        \abs b \is& \abs\newKPop(\abs\pid, \abs\kont, \abs q)
        % \abs b \is& \abs\newKaddr(\abs\pid,\larg{\ell}{i+1},\abs\timest)
    }\\\smallsep
    \ResCell{
        \abs\pi' &= \abs\pi\join
            [\abs\pid\mapsto\makeset{\tuple{\larg{\ell}{i+1},\abs\rho',\abs b,\abs t\:}}]\\
        \abs\sigma' &= \abs\sigma\join
            [\abs b\mapsto\makeset{
                \kArg{i+1}{\ell,\lst[0]{\abs d}{i},\abs\rho',\abs c}}]
    }\\\midrule

    \AbsRuleName{Apply}\label{absrule:apply}\\
    \CondCell{
        \abs\pi(\abs\pid)\ni& \abs q = \tuple{v,\abs\rho,\abs a,\abs t\:},\,
        \largn{\ell} = n\\
        \abs\sigma(\abs a) \ni&
            \kArg{n}{\ell,\lst[0]{\abs d}{n-1},\abs\rho',\abs c}\\
        \abs d_0 =& (\erl{fun(}\lst{x}{n}\erl{)}\to e, \abs\rho_0)\quad
        \abs d_n \is (v, \abs\rho)\\
        \abs\data_i \in& \abs\resolve(\abs\sigma,\abs d_i,)\\
        \abs b_i\is&\abs\newVaddr(\abs\pid,x_i,\abs \data_i,\abs q\:)\\
        \abs \rho''\is&\abs\rho'[x_1\mapsto \abs b_1\ldots x_n\mapsto \abs b_n]
        % \abs b_i\is&\abs\newVaddr(\abs\pid,x_i,\abs d_i,\abs\sigma,\abs\timest)
    }\\\smallsep
    \ResCell{
    \abs\pi' &= \abs\pi\join
        [\abs\pid\mapsto\makeset{
            \tuple{e,\abs\rho'',\abs c\:,\abs \tick(l,\abs t\:)}}]\\
    \abs\sigma' &= \abs\sigma\join
        [\abs b_1\mapsto\makeset{ \abs d_1}\ldots \abs b_n\mapsto\makeset{\abs d_n}]\\
        \vspace{-0.6725mm}%\\[1pt]
    }\\\bottomrule
\end{tabular}
\midspace%
\setlength{\colwd}{6.3cm}
\begin{tabular}[t]{@{}l@{}}
    \toprule
    \AbsRuleName{Vars}\label{absrule:vars}\\[2pt]
    \CondCell{
        &\abs\pi(\abs\pid)\ni\tuple{x,\abs\rho,\abs a,\abs t\:}\\
        &\abs\sigma(\abs\rho(x))\ni(v,\abs\rho')
    }\\\smallsep
    \ResCell{
        \abs\pi'&=\abs\pi\join[\abs\pid\mapsto\makeset{\tuple{v,\abs\rho',\abs a,\abs t\:}}]
    }\\[0.5em]
    \midrule[\heavyrulewidth]
    \RulesSection{Abstract communication}\\
    \midrule
    \AbsRuleName{Receive}\label{absrule:receive}\\[2pt]
    \CondCell{ %[when\\\null\quad]{
        &\abs\pi(\abs\pid)\ni \abs q = \tuple{
            e,
            \abs\rho,\abs a,\abs t\:}\\
        & e = \erl{receive}\ p_1 \to e_1\ldots p_n\to e_n\ \erl{end}\\
        &\abs\mailmatch(\lst{p}{n},\abs\mu(\abs\pid),\abs\rho,\abs\sigma) \ni
            (i,\abs\theta,\abs\mailbox)\\
        &\abs\theta = [x_1\mapsto\abs d_1\ldots x_k\mapsto\abs d_k]\\
        &\abs\data_j \in \abs\resolve(\abs\sigma,\abs d_j)\\
        &\abs b_j \is \abs\newVaddr(\abs\pid, x_j,\abs \data_j,\abs q\:)\\
        &\abs\rho'\is \abs\rho
            [x_1\mapsto\abs b_1\ldots x_k\mapsto\abs b_k]
    }\\\smallsep
    \ResCell{
        \abs\pi' &= \abs\pi\join
            [\abs\pid\mapsto\makeset{\tuple{e_i, \abs\rho',\abs a,\abs t\:}}]\\
        \abs\mu' &= \abs\mu
            [\abs\pid\mapsto\abs\mailbox]\\
        \abs\sigma'&=\abs\sigma\join
            [\abs b_1\mapsto\makeset{\abs d_1}\ldots\abs b_k\mapsto\makeset{\abs d_k}]\\[3pt]
    }\\
    \midrule

    \AbsRuleName{Send}\label{absrule:dispatch}\label{absrule:send}\\[2pt]
    \CondCell{
        \abs\pi(\abs\pid) \ni& \tuple{v,\abs\rho,\abs a,\abs t\:}\\
        \abs\sigma(\abs a) \ni& \kArg{2}{\ell,\abs d,\abs\iota',\_,\abs c}\\
        \abs d =& (\erl{send}, \_ )
    }\\\smallsep
    \ResCell{
        \abs\pi' &= \abs\pi\join[\abs\pid \mapsto\makeset{\tuple{v, \abs\rho,\abs c,\abs t\:}}]\\
        \abs\mu' &= \abs\mu[\abs\pid'\mapsto\abs\enqueue((v,\abs\rho),\abs\mu(\abs\pid'))]\\[3pt]
        \vspace{-3.5mm}
    }\\
    \bottomrule
\end{tabular}
\midspace%
\setlength{\colwd}{4.3cm}
\begin{tabular}[t]{@{}l@{}}
    \toprule
    \RulesSection{Abstract process creation}\\
    \midrule

    \AbsRuleName{Spawn}\label{absrule:proccr}\label{absrule:spawn}\\[4pt]
    \CondCell{
        \abs\pi(\abs\pid) \ni & \tuple{\erl{fun()}\to e,\abs\rho,\abs a,\abs t\:}\\
        \abs\sigma(\abs a) \ni & \kArg{1}{\ell, \abs d, \abs\rho',\abs c}\\
        \abs d=& (\erl{spawn},\_)\\
        \abs{\pid'} \is \null&\abs\newPid(\abs\pid, \ell, \abs t\:)
    }\\\smallsep
    \ResCell[then\\\null\quad]{
        \abs\pi' &= \abs\pi\hspace{-0.175mm}\join\hspace{-0.175mm}\left[
        \begin{aligned}
            \abs\pid &\mapsto\makeset{\tuple{\abs\pid', \abs\rho', \abs c, \abs t\:}},\\
            \abs\pid'&\mapsto\makeset{\tuple{e,\abs\rho,\StopAddr, \abs t_0}}
        \end{aligned}\right]\\[10pt]
        \abs\mu' &= \abs\mu\join[\abs\pid'\mapsto\abs\emptymb\:]\\[5pt]
    }\\\midrule

    \AbsRuleName{Self}\label{absrule:self}\\[4pt]
    \CondCell{
        \abs\pi(\abs\pid)\ni\tuple{\erl{self()},\abs\rho,\abs a,\abs t\:}
    }\\\smallsep
    \ResCell{
        \abs\pi'=\abs\pi\join[\abs\pid\mapsto\makeset{\tuple{\abs\pid,\abs\rho,\abs a,\abs t\:}}]
    }\\[0.45em]
    \midrule[\heavyrulewidth]
    \RulesSection{Initial abstract state}\\
    \midrule
    \parbox[t]{\colwd}{\rule[12pt]{0pt}{0pt}%
    \RuleName{AbsInit}\label{absrule:init}
        The initial state associated with a program $\Prog$ is
        \[\abs s_{\Prog} \is
            \absf(s_\Prog) = \tuple{\abs\pi_0,\abs\mu_0,\abs\sigma_0}\]
        where
        $\begin{array}[t]{@{}r@{\;}l@{}}
            \abs\pi_0 &= [\abs\pid_0\mapsto \makeset{\tuple{\Prog,[],\StopAddr,\abs t_0}}]\\
            \abs\mu_0 &= [\abs\pid_0\mapsto {\abs\emptymb\:}]\\
            \abs\sigma_0 &= [\StopAddr\mapsto\{\kTop\}]
        \end{array}$\\[4pt]
        \vspace{-0.25mm}
    }\\
    \bottomrule
\end{tabular}
\endgroup

\nocaptionrule
\caption{\label{fig:abstractSemantics}
    \textbf{Rules defining the Abstract Semantics}.
    The tables describe the conditions under which a transition
    $\abs s=\tuple{\abs\pi,\abs\mu,\abs\sigma}
        \cfaTo
    \tuple{\abs\pi',\abs\mu',\abs\sigma'}=\abs s'$
    can fire; the primed versions of the components of the states are identical to
    the non-primed ones unless indicated otherwise in the ``then'' part of the
    corresponding rule.
    We write $\sqcup$ for the join operation of the appropriate domain.
}
\label{fig:abstract-rules}
\end{figure*}
\endgroup

%%%%%%%%%%%%%%%%%%%%%%%%%%%%%%%%%%%%%%%%%%%%%%%%%%%%%%%%%%%%%%%%%%%%%%%%%%%%%%%

\begin{theorem}[Soundness of Analysis]
    \label{thm:CFAsound}
    Given a sound abstraction of the basic domains,
    if $s\semTo s'$ and $\absf[cfa](s) \leq u$, then
    there exists $u' \in \abs\State$ such that
        $\absf[cfa](s') \leq u'$ and
        $u \cfaTo u'$.
\end{theorem}

\noindent
See Appendix~\ref{apx:proofCFAsound} for a proof of the Theorem.

Now that we have defined a sound abstract semantics we give sufficient
conditions for its computability.

\begin{theorem}[Decidability of Analysis]
If a given (sound) abstraction of the basic domains is finite, then the derived
abstract transition relation defined in Figure~\ref{fig:abstract-rules} is
finite; it is also decidable if the associated auxiliary operations (in
Definition~\ref{def:soundAbstraction}) are computable.

\end{theorem}
\begin{proof}
    The proof is by a simple inspection of the rules: all the individual rules are decidable and the state space is finite.
\end{proof}

\paragraph{A Simple Mailbox Abstraction}
\label{sec:mbabs}

Abstract mailboxes need to be finite too in order for the analysis to be
computable. By abstracting addresses (and data) to a finite set, values, and
thus messages, become finite too. The only unbounded dimension of a mailbox
becomes then the length of the sequence of messages. We then abstract mailboxes
by losing information about the sequence and collecting all the incoming
messages in an un-ordered set:
\[
    \MailboxAbs_{\text{set}}\is
        \tuple{
            \powerset(\abs\Value),
            \subseteq,
            \union,
            \absf[set],
            \abs\enqueue_{\text{set}},
            \emptyset,
            \abs\mailmatch_{\text{set}}
        }
\]
where the abstract version of $\enqueue$ is the insertion in the set, as easily
derived from the soundness requirement; the matching function is similarly
derived from the correctness condition: writing $\vect{p} = \lst{p}{n}$
\[\begin{array}{l}
    \absf[set](\mailbox)
        \is \{ \absf(d) \mid \exists i.\ \mailbox_i = d \}\qquad
    \abs\enqueue_{\text{set}}(\abs d, \abs\mailbox)
        \is \makeset{\abs d}\union \abs\mailbox\\
    \abs\mailmatch_{\text{set}}(\vect{p},\abs\mailbox,\abs\rho,\abs\sigma)
        \is \left\{(i,\abs\theta,\abs\mailbox) \left|
            \begin{aligned}
            \abs d     & \in\abs\mailbox,\\
            \abs\theta & \in\abs\match_{\abs\rho,\abs\sigma}(p_i, \abs d)
            \end{aligned}
        \right.\right\}
\end{array}\]

We omit the straightforward proof that this constitutes a sound abstraction.

\paragraph{Abstracting Data.}
We included data in the value addresses in the definition of $\AddrVar$, cutting
contours would have been sufficient to make this domain finite. A simple
solution is to discard the value completely by using the trivial data
abstraction $\Data_0 \is \makeset{\databot}$ which is sound. If more precision
is needed, any finite data-abstraction would do: the analysis would then be able
to distinguish states that differ only because of different bindings in their
frame.

We present here a data abstraction particularly well-suited to languages with
algebraic data-types such as \Lang: the abstraction $\depthD[\abs\sigma, D]{e}$
discards every sub-term of $e$ that is nested at a deeper level than a
parameter $D$.

\begin{align*}
    \depthD[\abs\sigma, 0]{(e,\abs\rho)} & \is \{\databot\}\quad
    \depthD[\abs\sigma, D+1]{(\erl{fun...},\abs\rho)} \is \{\databot\}\\
    \depthD[\abs\sigma, D+1]{(\erl{#c#}(\lst{x}{n}),\abs\rho)} &
        \is\left\{
            \erl{#c#}(\lst{\abs\data}{n})\;\left|\;
            \begin{aligned}
                &\abs d_i\in\abs\sigma(\abs\rho (x_i)),\\
                &\abs\data_i\in\depthD[\abs\sigma, D]{\abs d_i}
            \end{aligned}
            \right.\right\}\\
\end{align*}
where $\databot$ is a placeholder for discarded subterms.

An analogous D-deep abstraction can be easily defined for concrete
values and we use the same notation for both; we use the notation
$\depthD[D]{\data}$ for the analogous function on elements of $\Data$.

We define $\DataAbs_D=\tuple{\Data_D,\absf[D],\absresolve[D]}$ to be the
`depth-D' data abstraction where
\begin{align*}
    \Data_{D+1} &
        \is \makeset{\databot}\union
            \{ \erl{#c#}(\lst{\abs\data}{n})\mid\abs\data_i\in\Data_D \} \\
    {\absf[D](\data)} &\is \depthD[D]{\data}\qquad
    {\absresolve[D](\abs\sigma, \abs d)} \is \depthD[\abs\sigma, D]{\abs d}
\end{align*}

The proof of its soundness is easy and we omit it.

\paragraph{Abstracting Time.}
\label{sec:kCFA}

Let us now define a specific time abstraction that amounts to a concurrent
version of a standard \kCFA. A \kCFA\ is an analysis parametric in $k$, which is
able to distinguish dynamic contexts up to the bound given by $k$.
We proceed as in standard \kCFA\ by truncating contours at
length $k$ to obtain their abstract counterparts:
\[
    \Contour_k \is \textstyle\Union_{0\leq i\leq k} \ProgLoc^i  \qquad
    \absfK[t] (\ell_1\dots\ell_k\cons\cntr) \is \ell_1\dots\ell_k
\]

\label{sec:0CFA}

The simplest analysis we can then define is the one induced by the basic domains abstraction $\tuple{\Data_0, \Time_0, \abs\Mailbox_{\text{set}}}$.
With this instantiation many of the domains collapse in to singletons.
Implementing the analysis as it is would lead however to an exponential
algorithm because it would record separate store and mailboxes for each abstract
state. To get a better complexity bound, we apply a widening following the lines
of~\cite[Section~7]{VanHorn:10}: instead of keeping a separate store and
separate mailboxes for each state we can join them keeping just a global copy of
each. This reduces significantly the space we need to explore: the algorithm
becomes polynomial time in the size of the program (which is reflected in the size of $\ProgLoc$).

Considering other abstractions for the basic domains easily leads to exponential
algorithms; in particular, the state-space grows linearly wrt the size of
abstract data so the complexity of the analysis using $\Data_D$ is exponential
in $D$.

\paragraph{Dealing with open programs.}
\label{sec:openterms}
\label{sec:inputgrammar}

Often it is useful to verify an open expression where its input is taken from a regular set of terms (see~\cite{Ong:11}).
We can reproduce this in our setting by introducing a new primitive
\erl[morekeywords={choice}]{choice} that non-deterministically calls one of its arguments.
For instance, an interesting way of closing \erl{N} in Example~\ref{ex:reslock}
would be by binding it to \erl{any_num()}:
\begin{erlang*}[morekeywords={choice},emph={zero, succ}]
letrec ...
    any_num() = choice(fun() -> zero,
                        fun() -> {succ, any_num()}).
in C = cell_start(), add_to_cell(any_num(), C).
\end{erlang*}
Now the uncoverability of the state where more than one instance of
\erl{inc} is running the protected section would prove that mutual exclusion is
ensured for any number of concurrent copies of \erl{inc}.

% !TEX root = main.tex
\section{Generating the Actor Communicating System}
\label{sec:generateACS}

The CFA algorithm we presented allows us to derive a sound `flat'
representation of the control-flow of the program. The analysis takes into
account higher-order computation and (limited) information about
synchronization. Now that we have this rough scheme of the possible transitions,
we can `guard' those transitions with actions which must take place in their
correspondence; these guards, in the form of `receive a message of this form'
or `send a message of this form' or `spawn this process' cannot be modelled
faithfully while retaining decidability of useful verification problems, as
noted in Section~\ref{sec:acs}. The best we can do, while remaining sound, is to
relax the synchronization and process creation primitives with counting
abstractions and use the guards to restrict the applicability of the
transitions. In other words, these guarded (labelled) rules will form the
definition of an ACS that simulates the semantics of the input \Lang\ program.

\medskip

\noindent\emph{Terminology}.\label{term:actcomp} We identify a common pattern of the rules in Figure~\ref{fig:abstractSemantics}. In each rule \textbf{R}, the premise distinguishes an abstract pid $\abs\pid$ and an abstract process state $\abs{q} = \tuple{e, \abs\rho, \abs a, \abs t\:}$ associated with $\abs\pid\:$~i.e.~$\abs q\in \abs\pi(\abs\pid)$
and the conclusion of the rule associates a new abstract process state---call it $\abs q'$---with $\abs\pid$~i.e.~$\abs{q'} \in \abs\pi'(\abs\pid)$. Henceforth we shall refer to $(\abs\pid, \abs q, \abs q')$ as the \emph{active components} of the rule \textbf{R}.

\begin{definition}[Generated ACS]\label{def:genacs}\label{def:genvas}
    Given a \Lang\ program \Prog, a sound basic domains abstraction
          $\AbsInt=\tuple{ \TimeAbs, \MailboxAbs, \DataAbs}$
    and a sound data abstraction for messages
          $\MsgAbs=\tuple{\Msg,\absf[msg], \absresolve[\text{msg}]}$

    the \emph{Actor communicating system} generated by
                            \Prog, $\AbsInt$ and $\MsgAbs$ is defined as
\[
    \Acs_\Prog \is
        \tuple{\abs\Pid,\abs\ProcState,\Msg,R,
            \absf(\pid_0),
            \absf(\pi_0(\pid_0))}
\]
where $s_{\Prog} = \tuple{\pi_0,\mu_0,\sigma_0,\cntr_0}$
is the initial state (according to \ref{rule:init}) with
$\pi_0=[\pid_0 \mapsto \tuple{\Prog,[],\StopAddr,t_0}]$
and the rules in $R$ are defined by induction over the following rules.
\begin{asparaenum}
\item If $\abs s \cfaTo \abs s'$ is proved by rule \ref{absrule:funeval} or
\ref{absrule:argeval} or \ref{absrule:apply} with active components $(\abs\pid, \abs q, \abs q')$, then
\begin{equation}
	\eTau\abs\pid[\abs q-->\abs q'] \in R
	\tag{AcsTau}\label{acsrule:tau}
\end{equation}

\item If $\abs s \cfaTo \abs s'$ is proved by \ref{absrule:receive} with active components $(\abs\pid, \abs q, \abs q')$ where $\abs d=(p_i,\abs \rho')$ is the abstract message matched by $\abs\mailmatch$ and $\abs m \in \absresolve[\text{msg}](\abs\sigma, \abs d)$, then
\begin{equation}
	\eRec\abs\pid[\abs q--?\abs m-->\abs q'] \in R
	\tag{AcsRec}\label{acsrule:rec}
\end{equation}

\item If $\abs s \cfaTo \abs s'$ is proved by \ref{absrule:send} with active
components $(\abs\pid, \abs q, \abs q')$ where $\abs d$ is the
abstract value that is sent and
$\abs{m} \in \absresolve[\text{msg}](\abs\sigma, \abs d)$, then
\begin{equation}
	\eSend\abs\pid[
		\abs q --\abs{\pid'}!\abs m--> \abs q'
	] \in R
	\tag{AcsSend}\label{acsrule:send}
\end{equation}

\item If $\abs s \cfaTo \abs s'$ is proved by \ref{absrule:spawn}  with active component $(\abs\pid, \abs q, \abs q')$ where $\abs\pid'$ is the new abstract pid that is generated in the premise of the rule, which gets associated with the process state $\abs q''=\tuple{e,\abs\rho,\StopAddr}$ then
\begin{equation}
	\eSpawn\abs\pid[
		\abs q--v\abs{\pid'}.\abs q''-->\abs q'
	] \in R
	\tag{AcsSp}\label{acsrule:spawn}
\end{equation}
\end{asparaenum}
\end{definition}

As we will make precise later, keeping $\abs\Pid$ and
$\abs\ProcState$ small is of par\-a\-mount importance for the model checking of
the generated ACS to be feasible. This is the main reason why we keep the
message abstraction independent from the data abstraction: this allows us to
increase precision with respect to types of messages, which is computationally
cheap, and keep the expensive precision on data as low as possible. It is
important to note that these two `dimensions' are in fact independent and a more
precise message space enhances the precision of the ACS even when using
$\Data_0$ as the data abstraction.

In our examples (and in our implementation) we use a $\Data_D$ abstraction for
messages where $D$ is the maximum depth of the receive patterns of the program.

\begin{definition}
	\label{def:alphaACS}
The abstraction function
\[
    \absf[acs] \colon \State \to
        {(\abs\Pid\times(\abs\ProcState\dunion\Msg) \to {\mathbb N})}
\]
relating concrete states and states of the ACS is defined as
\[
\absf[acs](s)\is
    \begin{cases}
    (\abs\pid,\abs q) & \mapsto
        \big|\{ \pid \mid \absf(\pid)=\abs\pid,
                          \absf(\pi(\pid))= \abs q \}\big| \\[2mm]
    (\abs\pid,\abs m) & \mapsto
        \left|\left\{ (\pid,i) \left|
        \begin{aligned}
            &\absf(\pid)=\abs\pid,\\
            &\absf[msg](\resolve(\sigma, \mu(\pid)_i))=\abs m
        \end{aligned}
        \right.\right\}\right|
    \end{cases}
\]
where $s=\tuple{\pi,\mu,\sigma}$.
\end{definition}

It is important to note that most of the decidable properties of the generated
ACS are not even expressible on the CFA graph alone: being able to predicate on
the contents of the counters means we can decide boundedness, mutual exclusion
and many other expressive properties. The next example shows one simple way in
which the generated ACS can be more precise than the bare CFA graph.

\begin{example}[Generated ACS]
    Given the following program:

\begin{erlang*}[emph={get,set,init}]
letrec
server= fun() -> receive {init, P, X} ->
                    send(P, ok), do_serve(X)
                 end.
do_serve= fun(X) -> receive
                        {init, _, _} -> error;
                        {set , Y}    -> do_serve(Y);
                        {get , P}    -> send(P,X),
                                        do_serve(X);
                    end.
in  S = spawn(server), send(S, {init, self(), a}),
        receive ok -> send(S, {set, b}) end.
\end{erlang*}
    our algorithm would output the following ACS starting from `main':
    \footnote{Labels are abbreviated to unclutter the picture; for example
        \texttt{\{\atom{init},\databot,\databot\}} is abbreviated with \atom{init}}

\begin{tikzpicture}[
    x=1.3cm,y=1.1cm,
    ->,
    empty state/.style={
        rounded corners,
        draw=black,
        semithick,
        minimum size=10pt,
        font={\scriptsize},
        inner sep=3pt,outer sep=1pt
    },
    state/.style={
        empty state,
        execute at end node={\vphantom{Aghqfi}},
    },
    thick,
    br/.style={looseness=.7,bend right=25pt},
    bl/.style={looseness=.7,bend left=25pt},
    shorten >=1pt,
    lbl/.style={font={\small}},
    ghost/.style={draw=ghostcol,dotted},
    spawn/.style={color=blue!60!black},
    error/.style={fill=red!60},
    ]

\newcommand{\tauskip}{.5}

\node [state, label=left:{$\abs\pid_s\colon$}] (server) {server};
\node [state, right=.8 of server] (doserve) {do\_serve};
\node [state, right=.8 of doserve] (rec) {receive};
\node [state, right=.8 of rec,error] (error) {error};

\draw (server) edge node[lbl,above] {?\atom{init}} (doserve)
      (doserve) edge node[lbl,above]{$\abs\pid_0$!\atom{ok}} (rec)
      (rec) edge[in=-120,out=-40,looseness=5]
                   node[lbl,above right=0pt and 7pt] {?\atom{set}}
      (rec) edge node[lbl,above] {?\atom{init}} (error);

\node [state, label=left:{$\abs\pid_0\colon$}, above=.2 of server] (main) {main};
\node [empty state, right=1 of main] (q1) {};
\node [empty state, right=.8 of q1] (q2) {};
\node [empty state, right=.6 of q2] (q3) {};
\node [empty state, right=.8 of q3] (q4) {};

\draw (main) edge node[lbl,above,spawn] {$\nu\abs\pid_s$.server} (q1)
      (q1) edge node[lbl,above] {$\abs\pid_s$!\atom{init}} (q2)
      (q2) edge node[lbl,above] {?\atom{ok}} (q3)
      (q3) edge node[lbl,above] {$\abs\pid_s$!\atom{set}} (q4)
;

\end{tikzpicture}

    The error state is reachable in the CFA graph but not in its Parikh
    semantics: the token \atom{init} is only sent once and never after \atom{ok}
    is sent back to the main process. Once \atom{init} has been consumed in the
    transition from `server' to `do\_serve' the counter for it will remain set to zero forever.
\end{example}

\begin{theorem}[Soundness of generated ACS]
    \label{thm:ACSsound}
    For all choices of $\AbsInt$ and $\MsgAbs$,
    for all concrete states $s$ and $s'$,
    if $s\semTo s'$ and
       $\absf[acs](s)\leq\vec{v}$
    then there exists $\vec{v'}$
    such that $\absf[acs](s')\leq \vec{v'}$,
    and $\vec{v}\acsTo\vec{v'}$.
\end{theorem}

\noindent
    See Appendix~\ref{apx:ACSsound} for a proof of the Theorem.

\begin{corollary}[Simulation]
Let $\Acs_\Prog$ be the ACS derived from a given \Lang\ program \Prog. We have $\sem{\Acs_\Prog}$ simulates the semantics of $\Prog$: for each $\Prog$-run $s\semTo s_1 \semTo s_2 \semTo \dots$, there exists a $\sem{\Acs_\Prog}$-run $\vec{v} \acsTo \vec{v}_1 \acsTo \vec{v}_2 \acsTo \dots$ such that $\absf[acs](s)=\vec{v}$ and for all $i$, $\absf[acs](s_i)\leq \vec{v}_i$.
\end{corollary}

Simulation preserves all paths so reachability (and coverability) is preserved.

\begin{corollary}
	If there is no $\vec{v}\geq \absf[acs](s')$ such that
    $\absf[acs](s)\acsTo^*\vec{v}$ then $s\not\to^* s'$.
\end{corollary}

\newcommand{\pidc}{\ensuremath{\abs\pid_{\text{c}}}}
\newcommand{\pidi}{\ensuremath{\abs\pid_{\text{i}}}}

\begin{example}[ACS Generated from Example~\ref{ex:reslock}]
A (simplified) pictorial representation of the ACS generated by our procedure
from the program in Example~\ref{ex:reslock} (with the parametric entry point
of Section~\ref{sec:inputgrammar}) is shown in Figure~\ref{fig:acsgen}, using a
\zCFA\ analysis.
The three pid-classes correspond to the starting process $\abs\pid_0$ and the
two static calls of spawn in the program, the one for the shared cell process
$\pidc$ and the other, $\pidi$, for all the processes running \erl{inc}.

The first component of the ACS, the starting one, just spawns a shared cell and
an arbitrary number of concurrent copies of the third component; these actions
increment the counter associated with states `res\_free' and `inc$_0$'.
The second component represents the intended protocol quite closely; note that
by abstracting messages they essentially become tokens and do not have a payload
anymore.
The rules of the third component clearly show its sequential behaviour.
The entry point is $(\abs\pid_0,\text{cell\_start})$.

The VAS semantics is accurate enough in this case to prove mutual exclusion
of, say, state `inc$_2$', which is protected by locks.
Let's say for example that $n>0$ processes of pid-class $\pidi$ reached state
`inc$_1$'; each of them sent a \atom{lock} message to the cell; note that now
the message does not contain the pid of the requester so all these messages are
indistinguishable; moreover the order of arrival is lost, we just count them.
Suppose that $\pidc$ is in state `res\_free'; since the counter for \atom{lock}
is $n$ and hence not zero, the rule labeled with ?\atom{lock} is enabled;
however, once fired the counter for `res\_free' is zero and the rule is
disabled. Now exactly one \atom{ack} can be sent to the `collective' mailbox
of pid-class \pidi\ so the rule receiving the ack is enabled; but as long as it
is fired, the only \atom{ack} message is consumed and no other \pidi\ process
can proceed. This holds until the lock is released and so on. Hence only one
process at a time can be in state `inc$_2$'.
This property can be stated as a coverability problem: can inc$_2=2$ be
covered? Since the VAS semantics is given in terms of a VAS, the property is
decidable and the answer can be algorithmically calculated.
As we saw the answer is negative and then, by soundness, we can infer it holds
in the actual semantics of the input program too.
\end{example}

\begin{figure}[tb]
\colorlet{ghostcol}{black!70}
\begin{tikzpicture}[
    x=1.3cm,y=1.1cm,
    ->,
    state/.style={rounded corners,
                draw=black,
                semithick,
                minimum size=10pt,
                font={\scriptsize},
                execute at end node={\vphantom{Aghqfi}},
                inner sep=3pt,outer sep=1pt},
    thick,
    br/.style={looseness=.7,bend right=25pt},
    bl/.style={looseness=.7,bend left=25pt},
    shorten >=1pt,
    lbl/.style={font={\small}},
    ghost/.style={draw=ghostcol,dotted},
    spawn/.style={color=blue!60!black},
    ]
\newcommand{\tauskip}{.5}
\node [state,label=left:{$\abs\pid_0\colon$}] (start0) {cell\_start};
\node [state, right=\tauskip of start0] (start1) {res\_start};
\node [state, right=1.4 of start1] (start2) {sp\_inc};
\node [state, right=\tauskip of start2] (start-stop) {stop};

\draw (start0) edge node[lbl,above] {$\tau$} (start1)
      (start1) edge node[lbl,above,spawn]{$\nu\pidc$.res\_free} (start2)
      (start2) edge[in=120,out=40,looseness=5]
                   node[lbl,above,spawn]{$\nu\pidi$.inc$_0$}
      (start2) edge node[lbl,above] {$\tau$} (start-stop);

\node [state,label=left:{$\pidc\colon$}, below=1.1 of start0]
    (cell0) {res\_free};
\node [state, above right=.1 and .5 of cell0] (cell1) {ack};
\node [state, below right=.1 and .5 of cell1] (cell1b) {res\_locked};
\node [state, below right=.1 and 1 of cell1b] (cell2) {Res};
\node [state, above right=.1 and 1 of cell1b] (cell3) {cell};

\draw (cell0)  edge[bl] node[lbl,above left] {?\atom{lock}} (cell1)
      (cell1)  edge[bl] node[lbl,above]{\pidi!\atom{ack}} (cell1b)
      (cell1b) edge[br] node[lbl,below] {?\atom{req}} (cell2)
               edge[bl] node[lbl,below] {?\atom{unlock}} (cell0)
      (cell2)  edge[br] node[lbl,right] {$\tau$} (cell3)
      (cell3)  edge[br] node[lbl,above] {\pidi!\atom{ans}} (cell1b)
      (cell3)  edge[bl] node[lbl,below right] {$\tau$} (cell1b)
;

\node [state,label=left:{$\pidi\colon$}, below=1.1 of cell0]
    (inc0) {inc$_0$};
\node [state, right=.8 of inc0] (inc1) {inc$_1$};
\node [state, below=.5 of inc1] (inc2) {inc$_2$};
\node [state, right=.6 of inc2] (inc3) {inc$_3$};
\node [state, right=.5 of inc3] (inc4) {inc$_4$};
\node [state, above=.5 of inc4] (inc5) {inc$_5$};
\node [state, right=1 of inc5] (inc-stop) {stop};

\draw (inc0) edge node[lbl,above] {\pidc!\atom{lock}} (inc1)
      (inc1) edge node[lbl,left] {?\atom{ack}} (inc2)
      (inc2) edge node[lbl,above] {\pidc!\atom{req}} (inc3)
      (inc3) edge node[lbl,above] {?\atom{ans}} (inc4)
      (inc4) edge node[lbl,right] {\pidc!\atom{req}} (inc5)
      (inc5) edge node[lbl,above] {\pidc!\atom{unlock}} (inc-stop);
\end{tikzpicture}
    \caption{ACS generated by the algorithm from Example~\ref{ex:reslock}}
    \label{fig:acsgen}
\end{figure}
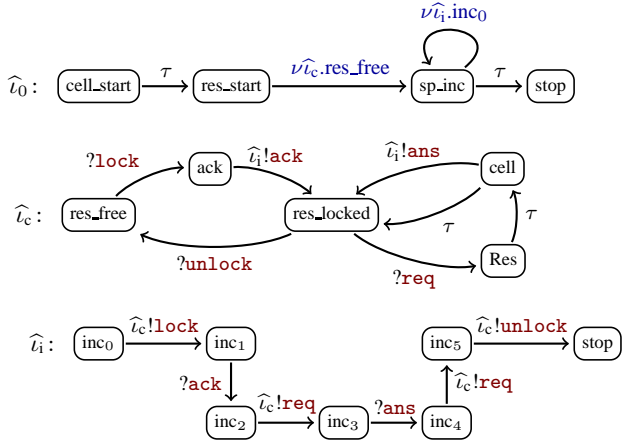

\paragraph{Complexity of the Generation.}
\label{sec:vasgen-compl}

Generating an ACS from a program amounts to calculating the analysis of Section~\ref{sec:cfa} and aggregating the relevant ACS rules for each transition of the analysis. Since we are adding \bigO{1} rules to $R$ for each transition, the complexity of the generation is the same as the complexity of the analysis itself. The only reason for adding more than one rule to $R$ for a single transition is the cardinality of $\Msg$ but since this costs only a constant overhead, increasing the precision with respect to message types is not as expensive as adopting more precise data abstractions.

\paragraph{Dimension of the Abstract Model.}
\label{sec:acsdim}

The complexity of coverability on VAS is \expspace\ in
the dimension of the VAS; hence for the approach to be practical,
it is critical to keep the number of components of the VAS underlying the
generated ACS small; in what follows we call \emph{dimension} of an
ACS the dimension of the VAS underlying its VAS semantics.

\begin{figure}[b]
\begin{center}
\begin{erlang}[emph={a,b}]
letrec no_a = fun(X)-> case X of a -> error; b -> ok end.
       send_b = fun(P)->send(P, b), send_a(P).
       send_a = fun(P)->send(P, a), send_b(P).
       stutter= fun(F)->receive _ -> unstut(F) end.
       unstut = fun(F)->receive X -> F(X), stutter(F) end.
in P = spawn(fun()->stutter(no_a)), send_a(P).
\end{erlang}
  \caption{A program that \Soter\ cannot verify because of the sequencing in mailboxes}
  \label{fig:stutter}
\end{center}
\end{figure}

Our algorithm produces an ACS with dimension $(|\abs\ProcState|+|\Msg|)\times
|\abs\Pid|$. With the \zCFA\ abstraction described at the end of
Section~\ref{sec:0CFA}, $\abs\ProcState$ is polynomial in the size of the
program and $\abs\Pid$ is linear in the size of the program so, assuming
$|\Msg|$ to be a constant, the dimension of the generated ACS is polynomial in
the size of the program, in the worst case.
Due to the parametricity of the abstract interpretation we can adjust for the right levels of precision and
speed. For example, if the property at hand is not sensitive to pids, one
can choose a coarser pid abstraction.
It is also possible to greatly reduce
$\abs\ProcState$: we observe that many of the control states
result from intermediate functional reductions; such reductions performed by different processes are independent, thanks to the actor model paradigm.
This allows for the use of preorder reductions. In our
prototype, as described in Section~\ref{sec:implementation}, we implemented a
simple reduction that safely removes states which only represent internal
functional transitions, irrelevant to the property at hand. This has proven to
be a simple yet effective transformation yielding a significant speedup. We
conjecture that, after the reduction, the cardinality of $\abs\ProcState$ is
quadratic only in the number of \erl{send}, \erl{spawn} and \erl{receive} of the
program.

% !TEX root = main.tex
\section{Evaluation, Limitations and Extensions}
\label{sec:discussion}\label{sec:evaluation}

\newcommand{\exprog}[1]{\textsf{#1}}

\paragraph{Empirical Evaluation.}
\label{sec:implementation}

%%%%%%%%%%%%%%%%%%%%%% TABLE WITH EXPERIMENTAL RESULTS %%%%%%%%%%%%%%%%%%%%%%%%

\begin{table*}[tb]
\newcommand{\HDR}{\textbf}
\newcommand{\s}{\phantom{*}}
\centering
\small
\begin{tabular}{lccccccrcccc}
\toprule

\multirow{2}{*}{\HDR{Example}}     &
\multirow{2}{*}{\HDR{LOC}}         &
\multirow{2}{*}{\HDR{PRP}}         &
\multirow{2}{*}{\HDR{SAFE?}}       &
\multicolumn{2}{c}{\HDR{ABSTR}}    &
\multicolumn{2}{c}{\HDR{ACS SIZE}} &
\multicolumn{4}{c}{\HDR{TIME}}     \\

                        &      &    &      &  D & M  & Places & Ratio & Analysis  & Simpl & \BFC   & Total   \\
\midrule
\exprog{reslock}        & 356  & 1  & yes  & 0  & 2  & 40   &  10\%  & 0.56  & 0.08  & 0.82   & 1.48   \\
\exprog{sieve}          & 230  & 3  & yes  & 0  & 2  & 47   &  19\%  & 0.26  & 0.03  & 2.46   & 2.76   \\
\exprog{concdb}         & 321  & 1  & yes  & 0  & 2  & 67   &  12\%  & 1.10  & 0.16  & 5.19   & 6.46   \\
\exprog{state\_factory} & 295  & 2  & yes  & 0  & 1  & 22   &   4\%  & 0.59  & 0.13  & 0.02   & 0.75   \\
\exprog{pipe}           & 173  & 1  & yes  & 0  & 0  & 18   &   8\%  & 0.15  & 0.03  & 0.00   & 0.18   \\
\exprog{ring}           & 211  & 1  & yes  & 0  & 2  & 36   &   9\%  & 0.55  & 0.07  & 0.25   & 0.88   \\
\exprog{parikh}         & 101  & 1  & yes  & 0  & 2  & 42   &  41\%  & 0.05  & 0.01  & 0.07   & 0.13   \\
\exprog{unsafe\_send}   & 49   & 1  & no\s & 0  & 1  & 10   &  38\%  & 0.02  & 0.00  & 0.00   & 0.02   \\
\exprog{safe\_send}     & 82   & 1  & no*  & 0  & 1  & 33   &  36\%  & 0.05  & 0.01  & 0.00   & 0.06   \\
\exprog{safe\_send}     & 82   & 4  & yes  & 1  & 2  & 82   &  34\%  & 0.23  & 0.03  & 0.06   & 0.32   \\
\exprog{firewall}       & 236  & 1  & no*  & 0  & 2  & 35   &  10\%  & 0.36  & 0.05  & 0.02   & 0.44   \\
\exprog{firewall}       & 236  & 1  & yes  & 1  & 3  & 74   &  10\%  & 2.38  & 0.30  & 0.00   & 2.69   \\
\exprog{finite\_leader} & 555  & 1  & no*  & 0  & 2  & 56   &  20\%  & 0.35  & 0.03  & 0.01   & 0.40   \\
\exprog{finite\_leader} & 555  & 1  & yes  & 1  & 3  & 97   &  23\%  & 0.75  & 0.07  & 0.86   & 1.70   \\
\exprog{stutter}        & 115  & 1  & no*  & 0  & 0  & 15   &  19\%  & 0.04  & 0.00  & 0.00   & 0.05   \\
\exprog{howait}         & 187  & 1  & no*  & 0  & 2  & 29   &  14\%  & 0.19  & 0.02  & 0.00   & 0.22   \\
\bottomrule
\end{tabular}

\smallskip
\nocaptionrule
\caption{
    \Soter\ Benchmarks.
    The number of lines of code refers to the compiled Core~Erlang.
    The PRP column indicates the number of properties which need to be proved.
    The columns D and M indicate the data and message abstraction depth
    respectively. In the ``Safe?'' column, ``no*'' means that the program satisfies
    the properties but the verification was inconclusive; ``no'' means
    that the program is not safe and \Soter\ finds a genuine counterexample.
    ``Places'' is the number of places of the underlying Petri net after the
    simplification; ``Ratio'' is the ratio of the number of places of the
    generated Petri net before and after the simplification.
    All times are in seconds.
}
\label{tab:benchmarks}
\end{table*}

%%%%%%%%%%%%%%%%%%%%%%%%%%%%%%%%%%%%%%%%%%%%%%%%%%%%%%%%%%%%%%%%%%%%%%%%%%%%%%%

To evaluate the feasibility of the approach, we have constructed \Soter, a prototype implementation of our method for verifying Erlang programs. Written in Haskell, \Soter\ takes as input a single Erlang module annotated with safety properties in the form of simple assertions. \Soter\ supports the full higher-order fragment and the (single-node) concurrency and communication primitives of Erlang; features not supported by \Soter\ are described in Remark~\ref{rem:LangVsCoreErlang}.
For more details about the tool see~\cite{Soter:AGERE}.
The annotated Erlang module is first compiled to Core Erlang by the Erlang compiler. A \zCFA-like analysis, with support for the $\Data_D$ data and message abstraction, is then performed on the compile; subsequently an ACS is generated. The ACS is simplified and then fed to the backend model-checker along with coverability queries translated from the annotations in the input Erlang program. \Soter's backend is the tool \BFC~\cite{Kaiser:12} which features a fast coverability engine for a variant of VAS.
At the end of the verification pathway, if the answer is YES then the program is safe with respect to the input property, otherwise the analysis is inconclusive.

In Table~\ref{tab:benchmarks} we summarise our experimental results.
Many of the examples are higher-order and use dynamic (and unbounded) process creation and non-trivial synchronization. Example~\ref{ex:reslock} appears as \exprog{reslock} and \Soter\ proves mutual exclusion of the clients' critical section.
\exprog{concdb} is the example program of~\cite{Huch:99} for which we prove mutual exclusion.
\exprog{pipe} is inspired by the `pipe' example of~\cite{Kobayashi:95}; the property proved here is boundedness of mailboxes. \exprog{sieve} is a dynamically spawning higher-order concurrent implementation of Erathostene's sieve inspired by a program by Rob Pike;%
\footnote{%
    see ``Concurrency and message passing in Newsqueak'',
    \url{http://youtu.be/hB05UFqOtFA}%
}
\Soter\ can prove all the mailboxes are bounded.

All example programs, annotated with coverability queries, can be viewed and verified using \Soter\  at \SoterUrl.

\paragraph{Limitations}
There are programs and properties that cannot be proved using any of the presented abstractions.
\begin{inparaenum}
    \item The program in Figure~\ref{fig:stutter} defines a simple function that
    discards a message in the mailbox and feeds the next to its functional
    argument and so on in a loop. Another process sends a `bad argument' and a
    good one in alternation such that only the good ones are fed to the
    function. The property is that the function is never called with a bad
    argument. This cannot be proved because sequential information of the
    mailboxes, which is essential for the verification, is lost in the counter
    abstraction.

    \item The program in Figure~\ref{fig:howait} defines a
    higher-order combinator that spawns a number of identical workers, each
    applied to a different task in a list. It then waits for all the workers to
    return a result before collecting them in a list which is subsequently
    returned. The desired property is that the combinator only returns when
    every worker has sent back its result. Unfortunately to prove this property,
    stack reasoning is required, which is beyond the capabilities of an ACS.
\end{inparaenum}

\begin{figure}[tb]
\begin{center}
\begin{erlang}[emph={req,reply,ans,wait_over,worker_finished,task1,task2}]
letrec
worker= fun(Task) -> ...

spawn_wait= fun(F, L) ->  spawn_wait'(F, fun()->[], L).
spawn_wait'= fun(F, G, L) ->
    case L of
        []     -> G();
        [T|Ts] ->
            S = self(),
            C = spawn(fun() ->
                        send(S, {ans, self(), F(T) })),
            F' = fun() ->
                    receive
                        {ans, C, R} -> [ R | G() ]
                    end,
            spawn_wait'(F, F', Ts)
    end.

in spawn_wait(worker, [task1, task2, ...]).
\end{erlang}
 \caption{A program that \Soter\ cannot verify because of the stack}
 \label{fig:howait}
\end{center}
\end{figure}

\paragraph{Refinement and Extensions.}
\label{sec:refinement}\label{sec:cegar}

Our parametric definition of the abstract semantics allows us to tune the
precision of the analysis when the abstraction is too coarse for the property to
be proved. For safety properties, the counter-example witnessing a no-instance
is a finite run of the abstract model. We conjecture that, given a spurious
counter-example, it is possible to compute a suitable refinement of the basic
domains abstraction so that the counter-example is no longer a run of the
corresponding abstract semantics. However a na\"{i}ve implementation of the
refinement loop would suffer from state explosion. A feasible CEGAR loop will
need to utilise sharper abstractions: it is possible for example to pinpoint a
particular pid or call or mailbox for which the abstract domains need to be more
precise while coarsely abstracting the rest. The development of a fully-fledged
CEGAR loop is a topic of ongoing research.

The general architecture of our approach, combining static analysis and abstract
model generation, can be adapted to accommodate different language features and
different abstract models. By appropriate decoration of the analysis, it is
possible to derive even more complex models for which semi-decision verification
procedures have been developed~\cite{Bouajjani:03,Long:12}.

% !TEX root = main.tex
\section{Related Work}\label{sec:relatedwork}

\paragraph{Static Analysis.}

Verification or bug-finding tools for Erlang%
~\cite{Marlow:97,Nystrom:03,Lindahl:06,Christakis:10,Christakis:11,Carlsson:06}
typically rely on static analysis.
The information obtained, usually in the form of a call graph, is then used to extract type constraints or infer runtime properties.
Examples of static analyses of Erlang programs in the literature include data-flow~\cite{Carlsson:06}, control-flow~\cite{Nystrom:03,Lindahl:06} and escape~\cite{Christakis:10} analyses.

Van Horn and Might~\cite{Might:11} derive a CFA for a multithreaded extension of
Scheme, using the same methodology~\cite{VanHorn:10} that we follow. The
concurrency model therein is thread-based, and uses a compare-and-swap
primitive. Our contribution, in addition to extending the methodology to Actor
concurrency, is to use the derived parametric abstract interpretation to
bootstrap the construction of an infinite-state abstract model for automated
verification.

Reppy and Xiao \cite{Reppy:07} and Colby \cite{Colby:95} analyse the channel
communication patterns of Concurrent ML (CML). CML is based on typed channels
and synchronous message passing, unlike the Actor-based concurrency model of
Erlang.

Venet~\cite{Venet:96} proposed an abstract interpretation framework for the sanalysis of $\pi$-calculus, later extended to other process algebras by Feret~\cite{Feret:05} and applied to CAP, a process calculus based on the Actor model, by Garoche~\cite{Garoche:06}. In particular, Feret's non-standard semantics can be seen as an alternative to Van Horn and Might's methodology, but tailored for process calculi.

\paragraph{Model Checking.}
Huch \cite{Huch:99} uses abstract interpretation and model checking to verify LTL-definable properties of a restricted fragment of Erlang programs:
\begin{inparaenum}
	\item order-one
	\item tail-recursive (subsequently relaxed in a follow-up paper~\cite{Huch:02}),
	\item mailboxes are bounded
	\item programs spawn a fixed, statically computable, number of processes.
\end{inparaenum}
Given a data abstraction function, his method transforms a program to an abstract, finite-state model; if a path property can be proved for the abstract model, then it holds for the input Erlang program.
In contrast, our method can verify Erlang programs of every finite order, with no restriction on the size of mailboxes, or the number of processes that may be spawned. Since our method of verification is by transformation to a \emph{decidable infinite-state system} that simulates the input program, it is capable of greater accuracy.

{McErlang} is a model checker for Erlang programs developed by Fredlund and Svensson~\cite{Fredlund:07}. Given a program, a B\"uchi automaton, and an abstraction function, McErlang explores on-the-fly a product of an abstract model of the program and the B\"{u}chi automaton encoding a property.
When the abstracted model is infinite-state, McErlang's exploration may not terminate.
McErlang implements a fully-fledged Erlang runtime system, and it supports a substantial part of the language, including distributed and fault-tolerant features.

ACS can be expressed as processes in a suitable variant of CCS~\cite{Milner:80}.
Decidable fragments of process calculi have been used in the literature to
verify concurrent systems. Meyer~\cite{Meyer:08} isolated a rich fragment of the
$\pi$-calculus called \emph{depth-bounded}. For certain patterns of
communication, this fragment can be the basis of an abstract model that avoids
the ``merging'' of mailboxes of the processes belonging to the same pid-class.
Erlang programs however can express processes which are not depth bounded.
We plan to address the automatic abstraction of arbitrary Erlang programs as depth-bounded process elsewhere.

\paragraph{Bug finding.}
Dialyzer~\cite{Lindahl:06,Christakis:10,Christakis:11} is a popular bug finding tool, included in the standard Erlang / OTP distribution.
Given an Erlang program, the tool uses flow and escape~\cite{Park:92} analyses to detect specific error patterns.
Building on top of Dialyzer's static analysis, \emph{success types} are derived. Lindahl and Sagonas' success types~\cite{Lindahl:06} `never disallow the use of a function that will not result in a type clash during runtime' and thus never generate false positives. Dialyzer puts to good use the type annotations that programmers do use in practice; it scales well and is effective in detecting `discrepancies' in Erlang code. However, success typing cannot be used to \emph{verify} program correctness.

% !TEX root = main.tex

\paragraph{Conclusion.} We have defined a generic analysis for \Lang, and a way of extracting from the analysis a simulating infinite-state abstract model in the form of an ACS, which can be automatically verified for coverability: if a state of the abstract model is not coverable then the corresponding concrete states of the input \Lang\ program are not reachable. Our constructions are parametric on the abstractions for \Time, \Mailbox\ and \Data, thus enabling different analyses (implementing varying degrees of precision with different complexity bounds) to be easily instantiated. In particular, with a \zCFA-like specialisation of the framework, the analysis and generation of the ACS are computable in polynomial time. Further, the dimension of the resulting ACS is polynomial in the length of the input \Lang\ program, small enough for the verification problem to be tractable in many useful cases. The empirical results using our prototype implementation \Soter\ are encouraging. They demonstrate that the abstraction framework can be used to prove interesting safety properties of non-trivial programs automatically. We believe that the proposed technique can easily be adapted to accommodate other languages and other abstract models. The level of generality at which the algorithm is defined seems to support the definition of a CEGAR loop readily, the formalisation of which is a topic for future work.

\bibliographystyle{plain}
\bibliography{refs}

\clearpage
\appendix
% !TEX root = ../main.tex
\section{Abstract Domains, Orders, Abstraction Functions and Abstract Auxiliary Functions}\label{apx:abstractdomains}

\subparagraph{Abstract Domains, Orders and Abstraction Functions:}
\begin{align*}
    %Abstract Pids
    \abs\Pid
         &\is \ProgLoc \times \abs\Time \\
         &\begin{aligned}
             \leqx[pid]\; &\is\;\; = \times \leqx[t]\\
             \absf[pid]\, &\is\; \id \times \absf[t]
         \end{aligned}\\
    %Abstract VAddr
    \abs\AddrVar &\is \abs\Pid \times \Var \times \abs\Data \times \abs\Time\\
        &\begin{aligned}
            \leqx[va]\; &\is\; \leqx[pid] \times = \times \leqx[d] \times \leqx[t]\\
            \absf[va]\; &\is \absf[pid]\; \times \id \times\; \absf[d]\; \times \absf[t]
        \end{aligned}\\
    %Abstract Env
    \abs\Env
         &\is \Var \finmap \abs\AddrVar\\
         &\begin{aligned}
             &\abs\rho \leqx[env] \abs\rho' \iff \forall x \in \Var\,.\, \abs\rho(x) \leqx[va] \abs\rho'(x)\\
             &\absf[env](\rho)(x) \is \absf[va](\rho(x))
         \end{aligned}\\
    %Abstract KAddr
    \abs\AddrKont &\is \abs\Pid \times \ProgLoc \times \abs\Env \times \abs\Time\\
        &\begin{aligned}
            \leqx[ka]\; &\is\; \leqx[pid] \times = \times \leqx[env] \times \leqx[t]\\
            \absf[ka]\; &\is \absf[pid]\; \times \id \times\; \absf[env]\; \times \absf[t]
        \end{aligned}\\
    %Abstract Closure
    \abs\Closure
        &\is \ProgLoc \times \abs\Env\\
        &\begin{aligned}
            \leqx[cl]\; &\is\; \;=\; \times \leqx[env]\\
            \absf[cl]\; &\is\; \id\; \times \absf[env]
        \end{aligned}\\
    %Abstract Value
    \abs\Val &\is \abs\Closure \dunion \abs\Pid\\
        &\begin{aligned}
            \leqx[val]\; &\is\; \;\leqx[cl]\; + \leqx[pid]\\
            \absf[val]\; &\is\; \absf[cl]\; + \absf[pid]
        \end{aligned}\\
    %Abstract ProcState
    \abs\ProcState
        & \is (\ProgLoc\dunion\abs\Pid) \times \abs\Env \times \abs\AddrKont \times \Time\\
        &\begin{aligned}
            \leqx[ps]\; &\is (= + \leqx[pid]) \times \leqx[env] \times \leqx[ka] \times \leqx[t]\\
            \absf[ps]\; &\is (\id + \absf[pid]\,) \times \absf[env] \times \absf[ka] \times \absf[t]
        \end{aligned}\\
    %Abstract Procs
    \abs\Procs
        & \is \abs\Pid\to\powerset(\abs\ProcState)\\
        &\begin{aligned}
             &\abs\pi \leqx[proc] \abs\pi' \iff \forall \abs\pid \in \abs\Pid\,.\, \abs\pi(\abs\pid) \subseteq \abs\pi'(\abs\pid)\\
             &\absf[procs](\pi)(\abs\pid ) \is \{\absf[ps](\pi(\pid)) \mid \absf[pid](\pid) = \abs\pid\,\}
         \end{aligned}\\
    %Abstract Mailboxes
    \abs\Mailboxes &\is \abs\Pid\to\abs\Mailbox \\
        &\begin{aligned}
             &\abs\mu \leqx[ms] \abs\mu' \iff \forall \abs\pid \in \abs\Pid\,.\, \abs\mu(\abs\pid) \leqx[m] \abs\mu'(\abs\pid)\\
             &\absf[ms](\mu)(\abs\pid ) \is \bigsqcup\{\absf[m](\mu(\pid)) \mid \absf[pid](\pid) = \abs\pid\,\}
         \end{aligned}\\
    %Abstract Store
    \abs\Store
        & \is (\!\abs\AddrVar  \to \powerset(\abs\Val))
                        % \\&\phantom{\is\null}
                  \times
              (\abs\AddrKont \to \powerset(\abs\Kont))\\
        &\begin{aligned}
             &\begin{aligned}
                \abs\sigma \leqx[st] \abs\sigma' \iff  &\forall \abs b \in \abs\AddrVar\,.\, \abs\sigma(\abs b) \subseteq \abs\sigma'(\abs b)\\
                                                        &\forall \abs a \in \abs\AddrKont\,.\, \abs\sigma(\abs a) \subseteq \abs\sigma'(\abs a)
             \end{aligned}\\
             &\absf[st](\sigma)(\abs b) \is \{\absf[val](\sigma(b)) \mid \absf[va](b) = \abs b\,\}, \abs b \in \abs\AddrVar\\
             &\absf[st](\sigma)(\abs a) \is \{\absf[kont](\sigma(a)) \mid \absf[ka](a) = \abs a\,\}, \abs a \hspace{-0.5mm}\in \hspace{-0.5mm}\abs\AddrKont
         \end{aligned}\\
    %Abstract State
    \abs\State
        & \is \abs\Procs \times \abs\Mailboxes \times \abs\Store \\
        &\begin{aligned}
            \leqx\; &\is \leqx[procs] \times \leqx[ms] \times \leqx[st]\\
            \absf[cfa]\; &\is (\id + \absf[pid]\,) \times \absf[env]
        \end{aligned}
\end{align*}
where we write $f + g := \{ (x,x') \mid (x,x') \in f \text{ or } (x,x') \in g \}$.
\subparagraph{Abstract Auxiliary Functions:}
\begin{align*}
    \abs\newKPush & \colon \abs\Pid \times \abs\ProcState \to \abs\AddrKont\\
    \abs\newKPush & (\abs\pid,(\ell,\abs\rho,\_,\abs t\:))
        \is (\abs\pid,\larg{\ell}{0},\abs\rho,\abs t\:)\\[.5em]
    \abs\newKPop  & \colon \abs\Pid \times \abs\Kont \times \abs\ProcState \to \abs\AddrKont\\
    \abs\newKPop & (\abs\pid,\abs\kont,\tuple{\_,\_,\_,\abs t\:})
        \is (\abs\pid,\larg{\ell}{i+1},\abs\rho,\abs t\:)\\
            & \text{ where }
               \kont=\kArg{i}{\ell,\dots,\abs\rho,\_}\\[.5em]
    \abs\newVaddr & \colon \abs\Pid \times \Var \times \abs\Data \times \abs\ProcState
        \to \abs\AddrVar\\
    \abs\newVaddr & (\abs\pid,x,\abs\data,\tuple{\_,\_,\_,\abs t\:})
        \is (\abs\pid,x,\abs\data, \abs t\:)\\
    \abs\newPid & \colon \abs\Pid \times \ProgLoc \times \abs\Time \to \abs\Pid\\
    \abs\newPid & ((\ell',\abs\cntr'),\ell,\abs\cntr\:) \is
        (\ell, \abs\tick^*(\,\abs\cntr, \abs\tick(\ell',\abs\cntr'))
\end{align*}
\subparagraph{Concrete and Abstract Match Function:}
\begin{align*}
%Concrete match
    &\match_{\rho,\sigma}(p_i, (x,\rho')) = \match_{\rho,\sigma}(p_i, \sigma(\rho'(x)))\\
    &\match_{\rho,\sigma}(x, d) = \{ x \mapsto d\} \text{ if } x \notin \fundom(\rho)\\
    &\match_{\rho,\sigma}(x, d) = \{ x \mapsto d\} \text{ if } \match_{\rho',\sigma}(p', d) \neq \bot\\
    &\quad\text{ where } (p',\rho') = \sigma(\rho(x))\\
    &\match_{\rho,\sigma}(p, (t,\rho')) = \bigotimes_{1 \leq i \leq n}\match_{\rho,\sigma}(p_i, (t_i,\rho'))\\
    &\quad
       \begin{aligned}
            \text{where}&&p &= \erl{#c#(}\lstc{p}{n}\erl{)}\\
                        &&t &= \erl{#c#(}\lstc{t}{n}\erl{)}\\
                        &&\theta \otimes \theta' &= \bot \qquad\quad \text{ if } \exists x\,.\, \theta(x) \neq \theta'(x)\\
                        &&\theta \otimes \theta' &= \theta \union \theta' \quad \text{ otherwise }\\
                        &&\bigotimes_{\emptyset} &= []
        \end{aligned}\\
    &\match_{\rho,\sigma}(p, d) = \bot \qquad \text{ otherwise}\\[14pt]
%Abstract match
    &\abs\match_{\abs\rho,\abs\sigma}(p_i, (x,\abs\rho')) = \bigcup_{ \abs d \in \abs\sigma(\abs\rho'(x))} \abs\match_{\abs\rho,\abs\sigma}(p_i, \abs d\:)\\
    &\abs\match_{\abs\rho,\abs\sigma}(x, d) = \{ x \mapsto \abs d\:\}\\
    &\abs\match_{\abs\rho,\abs\sigma}(p, (t,\abs\rho')) = \mathop{\abs\bigotimes}_{1 \leq i \leq n}\abs\match_{\abs\rho,\abs\sigma}(p_i, (t_i,\abs\rho'))\\
    &\quad
       \begin{aligned}
            \text{if }  &p = \erl{#c#(}\lstc{p}{n}\erl{)} \text{ and }\\
                        &t = \erl{#c#(}\lstc{t}{n}\erl{)}\\
            \text{where } &\mathop{\abs\bigotimes} (\{ \Theta_i \hspace{-.33mm}\mid\hspace{-.33mm} 1 \leq i \leq n\})\hspace{-0.66mm}=\hspace{-0.66mm}\left\{
                            \theta \,\left|\,
                            \begin{aligned}
                                &\theta = \bigotimes_{1 \leq i \leq n} \theta_i, \theta \neq \bot,\\
                                &\theta_i \in \Theta_i, 1 \leq i \leq n
                            \end{aligned}
                            \right.\right\}
        \end{aligned}\\
    &\abs\match_{\abs\rho,\abs\sigma}(p, d) = \emptyset \qquad \text{ otherwise}\\
\end{align*}

\begin{lemma}\label{apx:pcfasound:lem1}
Suppose the concrete domain $C = A\finmap B$ of partial functions has abstract domain $\abs C = \abs A\to\powerset(\abs B)$ with
the induced order $\leq$ and abstraction function $\absf[C] : C \to \abs C$ as specified in \ref{sec:absmachine:partialmap}
    then for all $f \in C$ and for all $\absf[C](f) \leq \abs f$
    \begin{align}
        \forall a \in \fundom(f)\,.\, \absf[B](f(a)) \in \abs f(\absf[A](a)). \label{apx:pcfasound:lem1:eqn1}
    \end{align}
    Further suppose $f,f' \in C$ such that $f' = f[a_1\mapsto b_1,\ldots,a_n\mapsto b_n]$ and let $\abs f, \abs f' \in \abs C$ such that $\abs f' = \abs f \join [\abs a_1\mapsto \abs b_1,\ldots, \abs a_n\mapsto \abs b_n]$ with $\absf[C](f) \leq \abs f$ and
    $\absf[A](a_i) = \abs a_i$, $\absf[B](b_i) = \abs b_i$ for $i = 1,\ldots,n$ then
    \begin{align}
        \absf[C](f') \leq \abs f'.    \label{apx:pcfasound:lem1:eqn2}
    \end{align}
\end{lemma}
\begin{proof}
Let $f \in C$ and $\abs f \in \abs C$ such that $\absf[C](f) \leq \abs f$.
The definition of $\leq$ implies that for all $\abs a \in \abs A$
\[
	\absf[C](f)(\abs a) \subseteq \abs f(\abs a).
\]
Take $a \in A$ and fix $\abs a = \absf[A](a)$ then we obtain
\[
    \absf[C](f)(\absf[A](a)) \subseteq \abs f(\absf[A](a)).
\]
Expanding the definition of $\absf[C]$ yields
\[
    \{\absf[B](b_0) \mid (a_0,b_0) \in f, \absf[A](a_0) = \absf[A](a))\} \subseteq \abs f(\absf[A](a)).
\]
In particular $\absf[B](f(a)) \in \{\absf[B](b_0) \mid (a_0,b_0) \in f, \absf[A](a_0) = \absf[A](a))\}$ which yields what we set out to prove
\[
	\forall a \in \fundom(f)\,.\, \absf[B](f(a)) \in \abs f(\absf[A](a)).
\]

Turning to equation \ref{apx:pcfasound:lem1:eqn2} we want to show $\absf[C](f') \leq \abs f'$.
Let $\abs a \in \abs A$ then there are several cases to consider
\begin{asparaenum}
    \item $\absf[C](f')(\abs a) = \absf[C](f)(\abs a)$. Then since $\absf[C](f) \leq \abs f \leq \abs f'$ we have
        $\absf[C](f)(\abs a) \subseteq \abs f(\abs a) \subseteq \abs f'(\abs a)$.
    \item $\abs a = \absf[A](a_i)$ for some $1 \leq i \leq n$. Then $\abs a = \abs a_i$ and thus
        \begin{align*}
            \absf[C](f')(\abs a) & = \{ \absf[B](b_i) \} \union \{\absf[B](f(a)) | \absf[A](a) = \abs a, a \neq a_i\} \\
                                 &  \subseteq \{\abs b_i\} \union \absf[C](f)(\abs a) \\
                                 & \subseteq \{\abs b_i\} \union \abs f(\abs a) \subseteq \abs f'(\abs a)
        \end{align*}
    \item otherwise there does \emph{not} exist $(a,b) \in f'$ such that $\absf[A](a) = \abs a$ and hence $\absf[C](f')(a) = \emptyset$ which makes our claim trivially true.
\end{asparaenum}
We can thus conclude that $\absf[C](f') \leq \abs f'$.
\end{proof}
\begin{corollary} \label{apx:pcfasound:cor1}
Let $\pi \in \Procs$ and $\abs\pi \in \abs\Procs$ such that $\absf[proc](\pi) \leq \abs\pi$,
let $\sigma \in \Store$ and $\abs\sigma \in \abs\Store$ such that $\absf[st](\sigma) \leq \abs\sigma$ and
Let $\mu \in \Mailboxes$ and $\abs\mu \in \abs\Mailboxes$ such that $\absf[ms](\mu) \leq \abs\mu$ then
\begin{enumerate}
	\item $\forall \pid \in \Pid\,.\, \absf[ps](\pi (\pid)) \in \abs\pi(\absf[pid](\pid))$ \label{apx:pcfasound:cor1:i}
	\item $\forall b \in \AddrVar\,.\, \absf[val](\sigma(b)) \in \abs\sigma(\absf[va](b))$
	\item $\forall a \in \AddrKont\,.\, \absf[kont](\sigma(a)) \in \abs\sigma(\absf[ka](a))$ \label{apx:pcfasound:cor1:iii}
    \item $\forall \pid \in \Pid\,.\, \absf[m](\mu (\pid)) \leq \abs\mu(\absf[pid](\pid))$ \label{apx:pcfasound:cor1:iv}
    \item $\forall \pid \in \Pid\,.\, \forall x \in \Var\,.\, \forall \delta \in \Data\,.\,
           \forall q \in \ProcState\,.\,$ \label{apx:pcfasound:cor1:v}
           \[
                \absf[va](\newVaddr(\pid,x,\delta,q)) = \abs\newVaddr(\absf[pid](\pid),x,\absf[d](\delta),\absf[ps](q))
           \]

\end{enumerate}
\end{corollary}
\begin{proof}
    Cases \ref{apx:pcfasound:cor1:i} - \ref{apx:pcfasound:cor1:iii} follow directly from Lemma \ref{apx:pcfasound:lem1};
    it remains to show the claims of \ref{apx:pcfasound:cor1:iv} and \ref{apx:pcfasound:cor1:v}.
    \begin{asparaenum}
        \item[(iv)] By assumption $\absf[ms](\mu) \leq \abs\mu$ which implies that
            \[
                \absf[ms](\mu)(\absf[pid](\pid)) \leq \abs\mu(\absf[pid](\pid)) = \abs\mu(\abs\pid).
            \]
            Expanding $\absf[ms]$ then gives us that $\absf[m](\mu(\pid)) \leq \absf[ms](\mu)(\absf[pid](\pid))$,
            since $\absf[ms](\mu) = \lambda\abs\pid.\bigsqcup\{\absf[m](\mu(\pid)) \mid \absf[pid](\pid)=\abs\pid\}$,
            which allows us to conclude
            \[
                \absf[m](\mu(\pid)) \leq \abs\mu(\abs\pid).
            \]
        \item[(v)] The claim follows straightforwardly from expanding $\newVaddr$ and $\abs\newVaddr$:
            \begin{align*}
                \absf[va](\newVaddr(\pid, x, \delta, q)) &= (\absf[pid](\pid),x,\absf[d](\delta),\absf[t](t))\\
                & \abs\newVaddr(\absf[pid](\pid), x, \absf[d](\delta), \absf[ps](q))
            \end{align*}
            where $q = \tuple{e, \rho, a, t}$.
    \end{asparaenum}
\end{proof}

% !TEX root = ../main.tex
\section{Proof of Theorem~\ref{thm:CFAsound}}\label{apx:proofCFAsound}
\begin{proof}[Proof of Theorem~\ref{thm:CFAsound}]
The proof
is by a case analysis of the rule that defines the concrete transition $s \semTo s'$. For each rule, the transition in the concrete system can be replicated in the abstract transition system using the abstract version of the rule, with the appropriate choice of abstract pid, the continuation from the abstract store, the message from the abstract mailbox, etc.

Let $s = \tuple{\pi, \mu, \sigma}
        \to
    \tuple{\pi', \mu', \sigma'} = s'$
and $u = \tuple{\abs\pi, \abs\mu, \abs\sigma}$
such that $\absf[proc](\pi) \leq \abs\pi$, $\absf[mail](\mu) \leq \abs\mu$,
      and $\absf[st](\sigma) \leq \abs\sigma$.
We consider a number of rules for illustration.

\medskip

\noindent\emph{Case}: (\ref{rule:send}).
    {%
    We know that $s \semTo s'$ using rule \ref{rule:send}; we can thus assume
    \begin{align*}
        \pi(\pid) &= \hbox to 6.6cm{$\tuple{v, \rho, a, t}$\hfill}\\
        \sigma(a) &= \kArg{2}{\ell,d,\iota',\_,c}\\
        d         &= (\erl{send}, \_ )
        \intertext{and for $s'$}
        \pi' &= \pi[\pid\mapsto\tuple{v, \rho, c, t}]\\
        \mu' &= \mu[\pid'\mapsto\enqueue((v,\rho),\mu(\pid'))]\\
        \sigma' &= \sigma.
    \end{align*}
    As a first step we will examine $u$ and show that $u \cfaTo u'$ for some $u'$.
    For $\abs\pi$ and $\abs\sigma$, writing $\abs\pid := \absf[pid](\pid)$, Corollary \ref{apx:pcfasound:cor1} gives us
    \begin{align*}
        &\tuple{v, \abs\rho, \abs a, \abs t\:} \in \abs\pi(\abs\pid)\\
        &\kArg{2}{\ell, \abs d, \abs\pid',\_,\abs c} \in \abs\sigma(\abs a)
    \end{align*}
    where $\absf[env](\rho) = \abs\rho$, $\absf[ak](a) = \abs a$, $\abs d = \absf[val](d)$, $\abs t = \absf[t](t)$, $\abs \pid' = \absf[pid](\pid')$ and $\abs c = \absf[ka](c)$.
    Rule \ref{absrule:send} is now applicable and we can set
    \begin{align*}
        \abs\pi' &:= \abs\pi\join[\abs\pid \mapsto\abs q]\\
        \abs q\;\, &:= \tuple{v, \abs\rho,\abs c, \abs t\:}\\
        \abs\mu' &:= \abs\mu[\abs\pid'\mapsto\abs\enqueue((v,\abs\rho),\abs\mu(\abs\pid'))]\\
        u' &:= \tuple{\abs\pi', \abs\mu', \abs\sigma}.
    \end{align*}
    It follows from rule (\ref{absrule:send}) that $u \cfaTo u'$. It remains to show that $\absf[cfa](s') \leq u'$ which follows directly from
    \begin{inparaenum}
        \item $\absf[proc](\pi') \leq \abs\pi'$ and
        \item $\absf[ms](\mu') \leq \abs\mu'$.
    \end{inparaenum}
    \begin{enumerate}
        \item $\absf[proc](\pi') \leq \abs\pi'$
            follows immediately from Lemma \ref{apx:pcfasound:lem1} since $\abs\pid = \absf[pid](\pid)$ and $\absf[ps](q) = \abs q$.
        \item $\absf[ms](\mu') \leq \abs\mu'$.
            It is sufficient to show that $\absf[pid](\pid') = \abs\pid'$, which is immediate, and
            $\absf[m](\mu'(\pid')) \leq \abs\mu(\abs\pid)$. For the latter, since $\absf[env](\rho) = \abs\rho$, a sound basic domain abstraction gives us
            \begin{align*}
                \absf[m](\mu'(\pid'))   &= \absf[m](\enqueue((v,\rho), \mu(\pid'))) \\
                                        &\leq \abs\enqueue ((v, \abs\rho), \abs\mu(\abs\pid')) = \abs\mu(\abs\pid)
            \end{align*}
            provided we can show $\absf[m](\mu(\pid')) \leq \abs\mu(\abs\pid')$; the latter inequality follows Corollary \ref{apx:pcfasound:cor1}.
            Hence we can conclude $\absf[ms](\mu') \leq \abs\mu'$ which completes the proof of this case.
    \end{enumerate}
    }%

\medskip

\noindent\emph{Case}: (\ref{rule:receive}).
    In the concrete $s \semTo s'$ using the \ref{rule:receive}, hence we can make the following assumptions
    \begin{align*}
        \pi(\pid) &= \hbox to 6.6cm{$\tuple{\erl{receive}\ p_1 \to e_1\ldots p_n\to e_n\ \erl{end},\rho,a,t} =: q$\hfill}\\
        (i,\theta,\mailbox) &= \mailmatch(\lst{p}{n},\mu(\pid),\rho,\sigma) \\
        \theta &= [x_1\mapsto d_1\ldots x_k\mapsto d_k]\\
        b_j &= \newVaddr(\pid, x_j, \delta_j, q) \\
        \delta_{j} &= \resolve(\sigma, d_{j})
    \intertext{and for state $s'$}
        \pi' &= \pi[\pid \mapsto q']\\
        q'    &= \tuple{e_i, \rho', a, t}\\
        \rho'&= \rho[x_1 \mapsto b_1\ldots x_k \mapsto b_k]\\
        \mu' &= \mu[\pid \mapsto \mailbox]\\
        \sigma' &= \sigma[b_1 \mapsto d_1 \ldots b_k \mapsto d_k].
    \end{align*}
    As a first step we will look at $u$ to prove there there exists a $u'$
    such that $u \cfaTo u'$ using rule \ref{absrule:receive}.
    We can invoke Corollary \ref{apx:pcfasound:cor1}, since $\absf[proc](\pi) \leq \abs\pi$, to obtain
    \[
        \abs q \is \tuple{\erl{receive}\ p_1 \to e_1\ldots p_n\to e_n\ \erl{end}, \abs\rho,\abs a, \abs t\:}\in \abs\pi(\abs\pid)
    \]
    where we write $\abs\rho \is \absf[env](\rho)$, $\abs a \is \absf[ka](a)$, $\abs t = \absf[t](t)$ and $\abs\pid \is \absf[pid](\pid)$.
    Moreover Corollary \ref{apx:pcfasound:cor1} gives us
    \[
        \absf[m](\mu(\pid)) \leq \abs\mu(\abs\pid).
    \]
    as $\absf[ms](\mu) \leq \abs\mu$ and $\abs\pid = \absf[pid](\pid)$.
    Since the instantiation of the basic domains is sound and $\absf[st](\sigma) \leq \abs\sigma$ we then know that
    \[
        (i,\abs\theta,\abs\mailbox) \in \abs\mailmatch(\vect{p},\abs\mu(\abs\pid),\abs\rho,\abs\sigma)
    \]
    such that $\abs\theta = \absf[sub](\theta)$ and $\abs\mailbox \geq \absf[m](\mailbox)$.
    Turning to the substitution $\abs\theta$ we can see that
    \[
        \abs\theta =  [x_1\mapsto \abs d_1\ldots x_k \mapsto \abs d_k]
    \]
    where $\abs d_i = \absf(d_i)$ for $1 \leq i \leq k$.
    Appealing to the sound basic domain instantiation once more, noting that $\absf[st](\sigma) \leq \abs\sigma$,
    yields that for $j = 1,\ldots,k$ we have $\abs\delta_j := \absf[d](\delta_j) \in \resolve(\abs\sigma, \abs{d_{j}})$;
    to obtain new abstract variable addresses we can now set
    $\abs{b_j} \is \abs\newVaddr(\abs\pid, x_j, \abs{\delta_j}, \abs q\:)$.
    Rule \ref{absrule:receive} is applicable now; we make the following definitions
    \begin{align*}
        \abs\pi'    &\is \abs\pi\join[\abs\pid\mapsto \abs q']\\
        \abs q'     &\is \tuple{e_i, \abs\rho',\abs a, \abs t\:} \\
        \abs\rho'   &\is \abs\rho[x_1\mapsto\abs b_1\ldots x_k\mapsto\abs b_k]\\
        \abs\mu'    &\is \abs\mu[\abs\pid\mapsto \abs\mailbox]\\
        \abs\sigma' &\is \abs\sigma\join[\abs b_1\mapsto\abs d_1\ldots\abs b_k\mapsto\abs d_k]\\
        \abs u'     &\is \tuple{\abs\pi', \abs\mu', \abs\sigma', \abs\timest}
    \end{align*}
    and observe that $u \cfaTo u'$.
    It remains to show $\absf[cfa](s') \leq u'$ which follows directly if we can prove
    \begin{inparaenum}
        \item $\absf[proc](\pi') \leq \abs\pi'$,
        \item $\absf[ms](\mu') \leq \abs{\mu}'$ and
        \item $\absf[st](\sigma') \leq \abs\sigma'$.
    \end{inparaenum}
    \begin{enumerate}
        \item $\absf[proc](\pi') \leq \abs\pi'$. We note that by Corollary \ref{apx:pcfasound:cor1} we know
            $\abs b_i = \absf[va](b_i)$ as
            $\abs\pid = \absf[pid](\pid)$, $\abs\delta_i = \absf[d](\delta_i)$ and $\absf[proc](q) = \abs q$
             for $1 \leq i \leq n$. It follows that $\abs\rho' = \absf[env](\rho')$
            and hence $\abs q' = \absf[ps](q')$. Lemma \ref{apx:pcfasound:lem1} is now applicable, since $\abs\pid = \absf[pid](\pid)$,
            to give $\absf[proc](\pi') \leq \abs\pi'$.
        \item $\absf[ms](\mu') \leq \abs{\mu}'$. It is sufficient to show that $\abs\pid = \absf[pid](\pid)$, which is immediate,
            and $\absf[m](\mailbox) \leq \abs\mailbox$ which we have already established above; hence we can conclude $\absf[ms](\mu') \leq \abs{\mu}'$.
        \item $\absf[st](\sigma') \leq \abs\sigma'$. The observation that $\abs b_i = \absf[va](b_i)$ and $\abs d_i = \absf[val](d_i)$
            allows the application of Lemma \ref{apx:pcfasound:lem1} which gives $\absf[st](\sigma') \leq \abs\sigma'$ as desired.
    \end{enumerate}
    This completes the proof of this case.

\medskip

\noindent\emph{Case}: (\ref{rule:apply}).
{%
    Since $s \semTo s'$ using rule \ref{rule:apply} we can assume that
    \begin{align*}
        \pi(\pid) &= \tuple{v, \rho, a, t}   =: q&                         \\
        \sigma(a) &= \kArg{n}{\ell,\lst[0]{d}{n-1},\rho',c} := \kont\\
        \largn{\ell} &= n\\
        \ensuremath{d_0 &= (\erl{fun(}\lst{x}{n}\erl{)}\to e, \rho_0)}\\
        d_n &= (v, \rho)
        \intertext{and for $i = 1, \ldots, n$}
        \delta_i &= \resolve(\sigma, d_i) \\
        b_i &= \newVaddr(\pid, x_i, \delta_i, q)
        \intertext{additionally for the successor state $s'$}
        \pi' &= \pi[\pid\mapsto q']\\
        & \text{where } q' \is \tuple{e,\rho'[x_1\to b_1\ldots x_n\to b_n],c,\tick(\ell, t)}&\\
        \sigma' &= \sigma[b_1\mapsto d_1\ldots b_n\mapsto d_n]\\
        \mu' &= \mu
    \end{align*}
    As a first step we will examine $u$ and show that there exists a $u'$ such that
    $u \cfaTo u'$ using rule \ref{absrule:apply}.
    From Corollary \ref{apx:pcfasound:cor1}, since $\absf[proc](\pi) \leq \abs\pi$, it follows that
    \[
        \abs q := \tuple{v, \absf[env](\rho), \absf[ka](a), \absf[t](t)} \in \abs\pi(\absf[pid](\pid)).
    \]
    Letting $\abs\rho := \absf[env](\rho)$, $\abs a := \absf[ka](a)$, $\abs t := \absf[t](t)$ and $\abs\pid := \absf[pid](\pid)$
    we can appeal to Corollary \ref{apx:pcfasound:cor1} again, as $\absf[st](\sigma) \leq \abs\sigma$, to obtain
    \[
    	\kArg{n}{\ell,\lst[0]{\abs d}{n-1},\abs\rho',\abs c} \in \abs\sigma(\abs a)
    \]
    where we write $\abs d_i := \absf[val](d_i)$ for $0 \leq i < n$, $\abs\rho' := \absf[env](\rho')$ and $\abs c := \absf[ka](c)$.
    Expanding $\absf[val]$ yields
    \[
        \abs d_0 = (\erl{fun(}\lst{x}{n}\erl{)}\to e, \abs\rho_0)
    \]
    where we write $\abs\rho_0 := \absf[env](\rho_0)$.
    Taking $\abs d_n := (v,\abs\rho)$ we obtain from
    our sound basic domain abstraction
    \[
        \abs \delta_i := \absf[d](\delta_i) \in \abs\resolve(\abs\sigma, \abs d_i) \text{ for } i = 1, \ldots, n
    \]
    as $\absf[st](\sigma) \leq \abs\sigma$ and  $\abs d_i = \absf[val](d_i)$.
    Turning to the abstract variable addresses we define
    \[
      \abs b_i \is \abs\newVaddr(\abs\pid,x_i,\abs\delta_i, \abs q\:) \text{ for } 1 \leq i \leq n.
    \]
    Rule \ref{absrule:apply} is now applicable and we define
    \begin{align*}
    	\abs\pi' &\is \abs\pi\join
		[\abs\pid\mapsto \abs q'
			]\\
        \abs q' &\is \tuple{e,\abs\rho'[x_1\to \abs b_1\ldots x_n\to \abs b_n],\abs c, \abs\tick(\ell,\abs t\:)}\\
	\abs\sigma' &\is \abs\sigma\join
		[\,\abs b_1\mapsto \abs d_1\ldots \abs b_n\mapsto \abs d_n]\\
		u' &\is \tuple{\abs\pi', \abs\mu, \abs\sigma'}.
    \end{align*}
    It is clear from rule \ref{absrule:apply} that $u \cfaTo u'$; it remains to show that $\absf[cfa](s') \leq u'$ to prove this case.
    The latter follows if we can justify
    \begin{inparaenum}
        \item $\absf[proc](\pi') \leq \abs\pi'$ and
        \item $\absf[st](\sigma') \leq \abs\sigma'$.
    \end{inparaenum}
    \begin{enumerate}
        \item $\absf[proc](\pi') \leq \abs\pi'$.
            We can appeal to Lemma \ref{apx:pcfasound:lem1} provided we can show that
            $\abs\pid = \absf[pid](\pid)$ and $\absf[ps](q') = \abs q'$
            where the former is immediate.
            For the latter, first observe that since we have a sound basic domain abstraction we know
            $\absf[t](\tick(\ell,t)) \leq \abs\tick(\ell, \abs t\:)$;
            however as $\Time$ is a flat domain so the above inequality is in fact an equality
            \[
                \absf[t](\tick(\ell, t)) = \abs\tick(\ell, \abs t\:).
            \]
            Moreover $\abs b_i = \absf[va](b_i)$ for $1 \leq i \leq n$ by Corollary \ref{apx:pcfasound:cor1}, hence
            \[
                \absf[env](\rho'[x_1\to b_1\ldots x_n\to b_n]) = \abs\rho'[x_1\to \abs b_1\ldots x_n\to \abs b_n]
            \]
            as $\absf[env](\rho') = \abs\rho'$; in combination with $\abs c = \absf[ka](c)$ we obtain the desired $\absf[ps](q') = \abs q'$. Thus we conclude that
            $\absf[proc](\pi') \leq \abs\pi'$.
        \item $\absf[st](\sigma') \leq \abs\sigma'$.
            Since $\absf[va](b_i) = \abs b_i$ and $\absf[val](d_i) = \abs d_i$ for $1 \leq i \leq n$ Lemma \ref{apx:pcfasound:lem1} is applicable once more and gives us $\absf[st](\sigma') \leq \abs\sigma'$ and completes the proof of this case.
    \end{enumerate}
}%
\end{proof}

% !TEX root = ../main.tex
\section{Proof of Theorem~\ref{thm:ACSsound}}\label{apx:ACSsound}
\subparagraph{Terminology.}
Analogously to our remark in section \ref{term:actcomp} on active components for rules \textbf{AbsR} in the abstract operational semantics it is possible to identify a similar pattern in the concrete operational semantics. Henceforth we will speak
of the concrete active component $(\pid,q,q')$ of a rule \textbf{R} of the concrete operational semantics and we will say
the abstract active component $(\abs\pid,\abs q,\abs q')$ of a rule \textbf{AbsR} of the abstract operational semantics where \textbf{AbsR} is the abstract counterpart of \textbf{R}. We will omit the adjectives abstract and concrete when there is no confusion.

\begin{lemma}\label{apx:ACSsound:lem1}
    Suppose $s \semTo s'$ using the concrete rule \textbf{R} with concrete active component $(\pid,q,q')$ and $\abs s \geq \absf[cfa](s)$.
    Then $\abs s \cfaTo \abs s'$ with $\abs s' \geq \absf[cfa](s')$ using rule $\textbf{AbsR}$ with abstract active component $(\absf[pid](\pid),\absf[ps](q),\absf[ps](q'))$.
\end{lemma}
\begin{proof}
    The claim follows from inspection of the proof of Theorem~\ref{thm:CFAsound}.
\end{proof}

\begin{proof}[Proof of Theorem~\ref{thm:ACSsound}]
    Suppose $s \semTo s'$ using rule \textbf{R} of the concrete operational semantics with active component $(\pid,q,q')$.
    We will prove our claim by case analysis on \textbf{R}.
    \begin{itemize}
        \item \textbf{R} = \ref{rule:funeval}, \ref{rule:argeval}, \ref{rule:apply} or \ref{rule:vars}.
            Take $\abs s = \absf[cfa](s)$; Lemma~\ref{apx:ACSsound:lem1} gives us that
            $\abs s \cfaTo \abs s'$ using abstract rule \textbf{AbsR} = \ref{absrule:funeval}, \ref{absrule:argeval}, \ref{absrule:apply} or \ref{rule:vars} respectively with active component $(\abs\pid, \abs q, \abs q')$ where $\abs\pid = \absf[pid](\pid)$, $\abs q = \absf[ps](q)$ and $\abs q' = \absf[ps](q')$.
            It follows that
            \[
                \vec{r} := \eTau\abs\pid[\abs q-->\abs q'] \in R.
            \]
            Since $s \semTo s'$ with active component $(\pid,q,q')$ it follows that
            \begin{enumerate}
                \item $\absf[acs](s)(\abs\pid, \abs q) \geq 1$,
                \item $\absf[acs](s')(\abs\pid, \abs q) = \absf[acs](s)(\abs\pid, \abs q) - 1$ and
                \item $\absf[acs](s')(\abs\pid, \abs q') = \absf[acs](s)(\abs\pid, \abs q') + 1$
            \end{enumerate}
            as $\abs\pid = \absf[pid](\pid)$, $\abs q = \absf[ps](q)$ and $\abs q' = \absf[ps](q')$.
            We know $\absf[acs](s)\leq\vec{v}$ and thus
            \[
                \vec{v}(\abs\pid, \abs q) \geq 1.
            \]
            If we define
            \[
              \vec{v'} := \vec{v}[(\abs\pid, \abs q) \mapsto \vec{v}(\abs\pid, \abs q) - 1, (\abs\pid, \abs q') \mapsto \vec{v}(\abs\pid, \abs q') + 1 ],
            \]
            then it is clear that
                $\vec{v} \acsTo \vec{v'}$ using rule $\vec{r} \in R$ and the inequalities
                \begin{align*}
                   \absf[acs](s')(\abs\pid, \abs q) = \absf[acs](s)(\abs\pid, \abs q) - 1 &\leq \vec{v}(\abs\pid, \abs q) - 1 = \vec{v'}(\abs\pid, \abs q)\\
                   \absf[acs](s')(\abs\pid, \abs q') = \absf[acs](s)(\abs\pid, \abs q') + 1 &\leq \vec{v}(\abs\pid, \abs q') + 1 = \vec{v'}(\abs\pid, \abs q');
                \end{align*}
                the consequence of the latter two is that $\absf[acs](s') \leq \vec{v'}$, since $\absf[acs](s)\leq\vec{v}$, which completes the proof of this case.
        \item \textbf{R} = \ref{rule:receive}.
            Letting $\abs s = \absf[cfa](s)$ Lemma~\ref{apx:ACSsound:lem1} yields that
            $\abs s \cfaTo \abs s'$ using abstract rule \ref{absrule:receive} with active component $(\abs\pid, \abs q, \abs q')$ where $\abs\pid = \absf[pid](\pid)$, $\abs q = \absf[ps](q)$ and $\abs q' = \absf[ps](q')$.
            We note that $s = \tuple{\pi,\sigma,\mu}$ and $\abs s = \tuple{\abs\pi,\abs\sigma,\abs\mu}$ where
            $\abs\pi = \absf[proc](\pi)$, $\abs\sigma = \absf[st](\sigma)$ and $\abs\mu = \absf[ms](\mu)$.
            Let the message matched by $\mailmatch$ and extracted from $\mu(\pid)$ be $d=(p_i, \rho')$
            then inspecting rule \ref{absrule:receive} we can assume that during $\abs s \cfaTo \abs s'$ message $\abs d=(p_i, \abs\rho')$, where $\abs\rho' = \absf[env](\rho')$, is matched by $\abs\mailmatch$.
            Since the message abstraction is a sound data abstraction we know that
            \[
                \abs m := \absf[msg](\resolve(\sigma, d))\in \absresolve[\text{msg}](\abs\sigma, \abs d\:)
            \]
            and hence we have
            \[
                \vec{r} := \eRec\abs\pid[\abs q--?\abs m-->\abs q'] \in R
            \]
            Additionally we know
            \begin{enumerate}
                \item $\absf[acs](s)(\abs\pid, \abs q) \geq 1$,
                \item $\absf[acs](s)(\abs\pid, \abs m) \geq 1$,
                \item $\absf[acs](s')(\abs\pid, \abs q) = \absf[acs](s)(\abs\pid, \abs q) - 1$,
                \item $\absf[acs](s')(\abs\pid, \abs q') = \absf[acs](s)(\abs\pid, \abs q') + 1$ and
                \item $\absf[acs](s')(\abs\pid, \abs m) = \absf[acs](s)(\abs\pid, \abs m) - 1$
            \end{enumerate}
            since $d$ is the message extracted from $\mu(\pid)$ and $\abs m = \absf[msg](\resolve(\sigma, d))$.
            By assumption we know $\absf[acs](s)\leq\vec{v}$ which implies
            \[
                \vec{v}(\abs\pid, \abs q) \geq 1 \text{ and } \vec{v}(\abs\pid, \abs m) \geq 1
            \]
            and so we can define
            \[
                \vec{v'} := \vec{v}\left[
                \begin{aligned}
                    (\abs\pid, \abs q) &\mapsto \vec{v}(\abs\pid, \abs q\:) - 1, \\
                    (\abs\pid, \abs m) &\mapsto \vec{v}(\abs\pid, \abs m) - 1, \\
                    (\abs\pid, \abs q') &\mapsto \vec{v}(\abs\pid, \abs q') + 1 \\
                \end{aligned}
                \right];
            \]
            it is then clear that, using rule $\vec{r} \in R$,
            $\vec{v} \acsTo \vec{v'}$ and
                \begin{align*}
                   \absf[acs](s')(\abs\pid, \abs q) = \absf[acs](s)(\abs\pid, \abs q) - 1 &\leq \vec{v}(\abs\pid, \abs q) - 1 = \vec{v'}(\abs\pid, \abs q)\\
                   \absf[acs](s')(\abs\pid, \abs m) = \absf[acs](s)(\abs\pid, \abs m) - 1 &\leq \vec{v}(\abs\pid, \abs m) - 1 = \vec{v'}(\abs\pid, \abs m)\\
                   \absf[acs](s')(\abs\pid, \abs q') = \absf[acs](s)(\abs\pid, \abs q') + 1 &\leq \vec{v}(\abs\pid, \abs q') + 1 = \vec{v'}(\abs\pid, \abs q').
                \end{align*}
            Hence, since $\absf[acs](s)\leq\vec{v}$, we can conclude $\absf[acs](s') \leq \vec{v'}$ as desired.
            \item \textbf{R} = \ref{rule:send}.
            Using Lemma~\ref{apx:ACSsound:lem1}, with $\abs s = \absf[cfa](s)$, gives
            $\abs s \cfaTo \abs s'$ with active component $(\abs\pid, \abs q, \abs q')$ for the abstract rule \ref{absrule:send}  where $\abs\pid = \absf[pid](\pid)$, $\abs q = \absf[ps](q)$ and $\abs q' = \absf[ps](q')$.
            Examining the concrete and abstract states we see $s = \tuple{\pi,\sigma,\mu}$ and $\abs s = \tuple{\abs\pi,\abs\sigma,\abs\mu}$ where
            $\abs\pi = \absf[proc](\pi)$, $\abs\sigma = \absf[st](\sigma)$ and $\abs\mu = \absf[ms](\mu)$.
            Let the pid of the recipient be $\pid'$ and let $d$ be the value enqueued to $\pid'$'s mailbox $\mu(\pid')$;
            inspecting the proof of Theorem \ref{thm:CFAsound}
            the pid of the abstract recipient is $\abs\pid' \is \absf[pid](\pid')$ and the sent abstract value is $\abs d = \absf[val](d)$.
            Appealing to the soundness of the message abstraction we obtain
            \[
                \abs m := \absf[msg](\resolve(\sigma, d))\in \absresolve[\text{msg}](\abs\sigma, \abs d\:)
            \]
            and hence we have
            \[
                \vec{r} := \eSend\abs\pid[\abs q --\abs{\pid'}!\abs m--> \abs q'] \in R
            \]
            Additionally we know
            \begin{enumerate}
                \item $\absf[acs](s)(\abs\pid, \abs q) \geq 1$,
                \item $\absf[acs](s')(\abs\pid, \abs q) = \absf[acs](s)(\abs\pid, \abs q) - 1$,
                \item $\absf[acs](s')(\abs\pid, \abs q') = \absf[acs](s)(\abs\pid, \abs q') + 1$ and
                \item $\absf[acs](s')(\abs\pid', \abs m) = \absf[acs](s)(\abs\pid', \abs m) + 1$
            \end{enumerate}
            since $d$ is the message enqueued to $\mu(\pid')$ and $\abs m = \absf[msg](\resolve(\sigma, d))$.
            From our assumption we know $\absf[acs](s)\leq\vec{v}$ and thus
            \[
                \vec{v}(\abs\pid, \abs q) \geq 1;
            \]
            making the definition
            \[
                \vec{v'} := \vec{v}\left[
                \begin{aligned}
                    (\abs\pid, \abs q) &\mapsto \vec{v}(\abs\pid, \abs q\:) - 1, \\
                    (\abs\pid, \abs q') &\mapsto \vec{v}(\abs\pid, \abs q') + 1, \\
                    (\abs\pid', \abs m) &\mapsto \vec{v}(\abs\pid', \abs m) + 1 \\
                \end{aligned}
                \right];
            \]
            we observe that we are able to use rule $\vec{r} \in R$ to make the step
            $\vec{v} \acsTo \vec{v'}$. Further the inequalities
                \begin{align*}
                   \absf[acs](s')(\abs\pid, \abs q) = \absf[acs](s)(\abs\pid, \abs q) - 1 &\leq \vec{v}(\abs\pid, \abs q) - 1 = \vec{v'}(\abs\pid, \abs q)\\
                   \absf[acs](s')(\abs\pid, \abs q') = \absf[acs](s)(\abs\pid, \abs q') + 1 &\leq \vec{v}(\abs\pid, \abs q') + 1 = \vec{v'}(\abs\pid, \abs q')\\
                   \absf[acs](s')(\abs\pid', \abs m) = \absf[acs](s)(\abs\pid', \abs m) - 1 &\leq \vec{v}(\abs\pid', \abs m) - 1 = \vec{v'}(\abs\pid', \abs m).
                \end{align*}
            imply, since $\absf[acs](s)\leq\vec{v}$, that $\absf[acs](s') \leq \vec{v'}$ which concludes the proof of this case.
            \item \textbf{R} = \ref{rule:spawn}.
            Take $\abs s = \absf[cfa](s)$; Lemma~\ref{apx:ACSsound:lem1} gives us that
            $\abs s \cfaTo \abs s'$ using abstract rule \ref{absrule:spawn} with active component $(\abs\pid, \abs q, \abs q')$ where $\abs\pid = \absf[pid](\pid)$, $\abs q = \absf[ps](q)$ and $\abs q' = \absf[ps](q')$.
            We note that $s = \tuple{\pi,\sigma,\mu}$, $s' = \tuple{\pi',\sigma',\mu'}$ and $\abs s = \tuple{\abs\pi,\abs\sigma,\abs\mu}$ where
            $\abs\pi = \absf[proc](\pi)$, $\abs\sigma = \absf[st](\sigma)$ and $\abs\mu = \absf[ms](\mu)$.
            Further we can assume
            \begin{align*}
                \pi(\pid) &= \hbox to 6.6cm{$\tuple{\erl{fun()}\to e,\rho,a,\cntr}$\hfill}\\
                \sigma(a) &= \kArg{1}{\ell, d, \rho',c}\\
                d&= (\erl{spawn},\_)\\
                \pid'    &\is  \newPid(\pid, \ell, \timest)\\
                \pi'(\pid) &= \tuple{\pid', \rho', c,\cntr} = q' \\
                \pi'(\pid') &= \tuple{e,\rho,\StopAddr,\cntr_0} =: q''
            \end{align*}
            Noting that we are replicating the step $s \semTo s'$ in the abstract $\abs s \cfaTo \abs s'$,
            $\absf[cfa](s) = \abs s$ and $\absf[pid] \circ \newPid = \abs\newPid \circ \absf$ we can see that the
            new abstract pid created is $\abs\pid' = \absf[pid](\pid')$ together with its process state $\abs q'' = \absf[ps](q'')$.
            Hence we can conclude that
            \[
                \vec{r} := \eSpawn\abs\pid[\abs q--v\abs{\pid'}.\abs q''-->\abs q'] \in R
            \]
            and we observe that
            \begin{enumerate}
                \item $\absf[acs](s)(\abs\pid, \abs q) \geq 1$,
                \item $\absf[acs](s')(\abs\pid, \abs q) = \absf[acs](s)(\abs\pid, \abs q) - 1$,
                \item $\absf[acs](s')(\abs\pid, \abs q') = \absf[acs](s)(\abs\pid, \abs q') + 1$ and
                \item $\absf[acs](s')(\abs\pid', \abs q'') = \absf[acs](s)(\abs\pid', \abs q'') + 1$.
            \end{enumerate}
            Now the assumption $\absf[acs](s)\leq\vec{v}$ allows us to conclude
            \[
                \vec{v}(\abs\pid, \abs q) \geq 1;
            \]
            so that we can define
            \[
                \vec{v'} := \vec{v}\left[
                \begin{aligned}
                    (\abs\pid, \abs q) &\mapsto \vec{v}(\abs\pid, \abs q\:) - 1, \\
                    (\abs\pid, \abs q') &\mapsto \vec{v}(\abs\pid, \abs q') + 1, \\
                    (\abs\pid', \abs q'') &\mapsto \vec{v}(\abs\pid', \abs q'') + 1\\
                \end{aligned}
                \right];
            \]
            and use rule $\vec{r} \in R$ to make the step
            $\vec{v} \acsTo \vec{v'}$. Further with the inequalities
                \begin{align*}
                   \absf[acs](s')(\abs\pid, \abs q) = \absf[acs](s)(\abs\pid, \abs q) - 1 &\leq \vec{v}(\abs\pid, \abs q) - 1 = \vec{v'}(\abs\pid, \abs q)\\
                   \absf[acs](s')(\abs\pid, \abs q') = \absf[acs](s)(\abs\pid, \abs q') + 1 &\leq \vec{v}(\abs\pid, \abs q') + 1 = \vec{v'}(\abs\pid, \abs q')\\
                   \absf[acs](s')(\abs\pid', \abs q'') = \absf[acs](s)(\abs\pid', \abs q'') - 1 &\leq \vec{v}(\abs\pid', \abs q'') - 1 = \vec{v'}(\abs\pid', \abs q'').
                \end{align*}
            and our assumption $\absf[acs](s)\leq\vec{v}$ we see that $\absf[acs](s') \leq \vec{v'}$ which completes the proof of this case and the theorem.
    \end{itemize}
\end{proof}

\end{document}